\documentclass[
    reprint,
    aps,
    prx,
    longbibliography
]{revtex4-2}
\usepackage{graphicx}
\usepackage{dsfont}
\usepackage{xcolor}
\usepackage{amssymb}
\usepackage{amsmath}
\usepackage{mathtools}
\usepackage{amsthm}
\newtheorem{theorem}{Theorem}
\newtheorem{lemma}{Lemma}
\newtheorem{definition}{Definition}
\usepackage{physics}
\usepackage{hyperref}
\hypersetup{colorlinks=true, citecolor=blue, linktocpage=true, hypertexnames=false}

\makeatletter
\patchcmd{\@ssect@ltx}{\addcontentsline{toc}{#1}{\protect\numberline{}#8}}{}{}{}
\renewcommand{\fnum@figure}{\textbf{Figure \thefigure}}
\makeatother
\begin{document}

\title{On the Sample Complexity of Quantum Boltzmann Machine Learning}

\author{Luuk Coopmans}
\email{luukcoopmans2@gmail.com}
\affiliation{Quantinuum, Partnership House, Carlisle Place, London SW1P 1BX, United Kingdom}

\author{Marcello Benedetti}
\email{marcello.benedetti@quantinuum.com}
\affiliation{Quantinuum, Partnership House, Carlisle Place, London SW1P 1BX, United Kingdom}

\date{December 17, 2024}

\begin{abstract}
Quantum Boltzmann machines (QBMs) are machine-learning models for both classical and quantum data. We give an operational definition of QBM learning in terms of the difference in expectation values between the model and target, taking into account the polynomial size of the data set.
By using the relative entropy as a loss function this problem can be solved without encountering barren plateaus.
We prove that a solution can be obtained with stochastic gradient descent using at most a polynomial number of Gibbs states. 
We also prove that pre-training on a subset of the QBM parameters can only lower the sample complexity bounds. In particular, we give pre-training strategies based on mean-field, Gaussian Fermionic, and geometrically local Hamiltonians. We verify these models and our theoretical findings numerically on a quantum and a classical data set. 
Our results establish that QBMs are promising machine learning models.
\end{abstract}

\maketitle

\section*{Introduction}

Machine learning (ML) research has developed into a mature discipline with applications that impact many different aspects of society. Neural network and deep learning architectures have been deployed for tasks such as facial recognition, recommendation systems, time series modeling, and for analyzing highly complex data in science. In addition, unsupervised learning and generative modeling techniques are widely used for text, image, and speech generation tasks, which many people encounter regularly via interaction with chatbots and virtual assistants. Thus, the development of new machine learning models and algorithms can have significant consequences for a wide range of industries, and more generally, society as a whole~\cite{Bommasani_2021}.

Recently, researchers in quantum information science have asked the question of whether quantum algorithms can offer advantages over conventional machine learning algorithms. This has led to the development of quantum algorithms for gradient descent, classification, generative modeling, reinforcement learning, as well as many other tasks~\cite{Dunjko_2018, Ciliberto_2018, Benedetti_2019, Lamata_2020, Cerezo_2022}. However, one cannot straightforwardly generalize results from the conventional ML realm to the quantum ML realm. One must carefully reconsider the data encoding, training complexity, and sampling in the quantum machine learning (QML) setting. For example, it is yet unclear how to efficiently embed large data sets into quantum states so that a genuine quantum speedup is achieved~\cite{Aaronson_2015, Hoefler_2023}. Furthermore, as quantum states prepared on quantum devices can only be accessed via sampling, one cannot estimate properties with arbitrary precision. This gives rise to new problems, such as barren plateaus~\cite{McClean_2018, Grant2019initialization, Cerezo_2021, Marrero_2021, CerveroMartin2023, rudolph2023trainability}, that make the training of certain QML models challenging or even practically impossible.

In this work, we show that a particular quantum generative model, the quantum Boltzmann machine~\cite{Amin_2018, Benedetti_2017, Kieferova_2017, Kappen_2020}~(QBM) without hidden units, does not suffer from these issues, and can be trained with polynomial sample complexity on future fault-tolerant quantum computers. QBMs are physics-inspired ML models that generalize the classical Boltzmann machines to quantum Hamiltonian ans{\"a}tze. The Hamiltonian ansatz is defined on a graph where each vertex represents a qubit and each edge represents an interaction. The task is to learn the strengths of the interactions, such that samples from the quantum model mimic samples from the target data set. Quantum generative models of this kind could find use in ML for science problems, by learning approximate descriptions of the experimental data. QBMs could also play an important role as components of larger QML models~\cite{Benedetti_2018, Khoshaman_2019, Crawford_2019, Wilson_2021, Wang_2021}. This is similar to how classical Boltzmann machines can provide good weight initializations for the training of deep neural networks~\cite{Hinton_2006}. One advantage of using a QBM over a classical Boltzmann machine is that it is more expressive since the Hamiltonian ansatz can contain more general non-commuting terms. This means that in some settings the QBM outperforms its classical counterpart, even for classical target data~\cite{Kappen_2020}. Here we focus on QBMs without hidden units since their inclusion leads to additional challenges~\cite{Wiebe_2019, Kieferova_2021, Marrero_2021}, and there exists no evidence that they are beneficial.

In order to obtain practically relevant results, we begin by providing an operational definition of QBM learning. Instead of focusing on an information-theoretic measure, we assess the QBM learning performance by the difference in expectation values between the target and the model. We require that the data set and model parameters can be efficiently stored in classical memory. This means that the number of training data points is polynomial in the number of features if the target is classical, or in the number of particles if the target is quantum. Thus, statistics computed from such data set have polynomial precision. We then employ stochastic gradient descent methods~\cite{Khaled2020, Garrigos2023} in combination with shadow tomography~\cite{Aaronson_2017, Huang_2021, Rouze_2023} to prove that QBM learning can be solved with polynomially many evaluations of the QBM model. Each evaluation of the model requires the preparation of one \emph{Gibbs state} and therefore in this paper we refer to the sample complexity as the required number of Gibbs state preparations. We also prove that the required number of Gibbs samples is reduced by pre-training on a subset of the parameters. This means that classically pre-training a simpler QBM model can potentially reduce the (quantum) training complexity. For instance, we show that one can analytically pre-train a mean-field QBM and a Gaussian Fermionic QBM in closed form. In addition, we show that one can pre-train geometrically local QBMs with gradient descent, which has some improved performance guarantees. To the best of our knowledge, this is the first time these exactly solvable models are used for either training or pre-training QBMs.

It is important to note that the time complexity of preparing and estimating properties of Gibbs states at finite temperature is largely unknown, but can be exponential in the worst case~\cite{Bravyi_2022}. For cases where the Gibbs state preparation is intractable, the polynomial sample complexity of QBM learning becomes irrelevant as obtaining one sample already takes exponential time. On the other hand, the polynomial sample complexity is a prerequisite for efficient learning which many other QML models do not share. In this work we are only concerned with the sample complexity of the QBM learning problem, and we do not discuss any specific Gibbs sampling implementation. We shall mention in passing that Gibbs states satisfying certain locality constraints can be efficiently simulated by classical means, e.g., using tensor networks~\cite{Alhambra_2021,Kuwahara2021}. Furthermore, generic Gibbs states can be prepared and sampled on a quantum computer by a variety of methods (e.g., see~\cite{Poulin_2009, Temme_2011, kastoryano2016quantum, Chowdhury_2016, anschuetz2019realizing, Holmes_2022, chifang2023quantum, zhang2023dissipative, Coopmans_2023}), potentially with a quadratic speedup over classical algorithms. There also exist hybrid quantum-classical variational algorithms for Gibbs state preparation~\cite{Wu_2019, chowdhury2020variational, Liu_2021, zhao2018bias, consiglio2023variational, huijgen2023training}, but there are currently many open issues regarding their trainability~\cite{McClean_2018, Cerezo_2021, Anschuetz_2022, rudolph2023trainability}. In practice QBM learning allows for great flexibility in model design and can be combined with almost any Gibbs sampler.

Throughout the paper we compare QBM learning and a related, but crucially different, problem called Hamiltonian learning~\cite{Wiebe_2014, franca2020, Anshu_2021, Haah_2024, Rouze_2023, Onorati_2023}. Learning the Hamiltonian of a physical system in thermal equilibrium can be thought of as performing quantum state tomography. This has useful applications, for example in the characterization of quantum devices, but requires prior knowledge about the target system and suffers an unfavorable scaling with respect to the temperature~\cite{Anshu_2021}. In contrast, the aim of QBM learning is to produce a generative model of the data at hand, with the Hamiltonian playing the role of an ansatz.

Our setup and theoretical results are summarized in Fig.~\ref{fig:intuition}. We provide classical numerical simulation results in Figs.~\ref{fig:noiseless_results} and~\ref{fig:noisy_results} which confirm our analytical findings. We then conclude the paper with a discussion of the generalization capability of QBMs and avenues for future work.

\begin{figure*}[t]
    \centering    
    \includegraphics[width=1.0\textwidth]{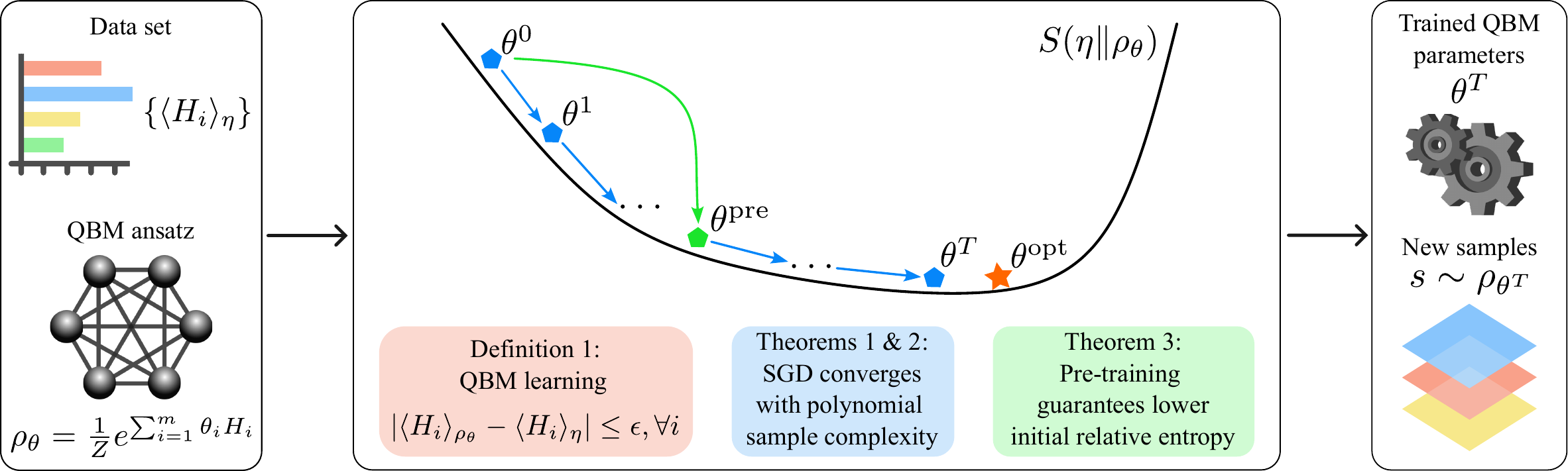}
    \caption{Summary of the results. The quantum Boltzmann machine (QBM) learning algorithm takes as input a data set of size polynomial in the number of features/qubits, and an ansatz with parameters $\theta$ and a set of $m$ Hermitian operators $\{H_i\}$. In Definition~\ref{def:problem} we provide an operational definition of the QBM learning problem where the model and target expectations must be close to within polynomial precision $\epsilon$. A solution $\theta^\mathrm{opt}$ is guaranteed to exist by Jaynes' principle. With Theorems~\ref{thm:training} and ~\ref{thm:training_alpha} we establish that QBM learning can be solved by minimizing the quantum relative entropy $S(\eta \| \rho_\theta)$ with respect to $\theta$ using stochastic gradient descent (SGD). This requires a polynomial number of iterations $T$, each using a polynomial number of Gibbs state preparations, i.e., the sample complexity is polynomial. With Theorem~\ref{thm:pretraining} we prove that pre-training strategies that optimize a subset $\theta^\mathrm{pre}$ of the QBM parameters are guaranteed to lower the initial quantum relative entropy. After training the QBM can be used to generate new synthetic data. Icons by \href{https://www.svgrepo.com/}{SVG Repo}, used under \href{https://creativecommons.org/publicdomain/zero/1.0/}{CC0 1.0} and adapted for our purpose.}
    \label{fig:intuition}
\end{figure*}

\section*{Methods}

We start by formally setting up the quantum Boltzmann machine (QBM) learning problem. We give the definitions of the target and model and describe how to assess the performance by the precision of the expectation values. These definitions and assumptions are key to our results, and we introduce and motivate them here carefully. To further motivate our problem definition we compare to other related problems in the literature, such as quantum Hamiltonian learning.  

We consider an $n$-qubit density matrix $\eta$ as the \emph{target} of the machine learning problem. If the target is classical, $n$ could represent the number of features, e.g., the pixels in black-and-white pictures, or more complex features that have been extracted and embedded in the space of $n$ qubits. If the target is quantum, $n$ could represent the number of  spin-$\tfrac{1}{2}$ particles, but again more complex many-body systems can be embedded in the space of $n$ qubits. In the literature, it is often assumed that algorithms have direct and simultaneous access to copies of $\eta$. We do not take that path here. Instead, we consider a setup where one only has access to \emph{classical} information about the target. We have a data set $\mathcal{D} = \{s^\mu\}$ of data samples $s^\mu$ and, crucially, we assume the data set can be efficiently stored in a classical memory: the amount of memory required to store each data sample is polynomial in $n$, and there are polynomially many samples. For example, $s^\mu$ can be bitstrings of length $n$; this includes data sets like binary images and time series, categorical and count data, or binarized continuous data. As another example, the data may originate from measurements on a quantum system. In this case, $s^\mu$ identifies an element of the positive operator-valued measure describing the measurement. 

Next, we define the machine learning \emph{model} which we use for data fitting. The fully-visible QBM is an $n$-qubit mixed quantum state of the form 
\begin{equation}
    \rho_\theta = \frac{e^{\mathcal{H}_\theta}}{Z}, 
\end{equation} 
where $Z=\Tr[e^{\mathcal{H}_\theta}]$ is the partition function. The parameterized Hamiltonian is defined as 
\begin{equation}
\label{eq:generic_H}
    \mathcal{H}_\theta = \sum_{i=1}^m \theta_i H_i,
\end{equation} 
where $\theta \in \mathbb{R}^m$ is the parameter vector, and $\{H_i\}$ is a set of $m$ Hermitian and orthogonal operators acting on the $2^n$-dimensional Hilbert space. For example, these could be $n$-qubit Pauli operators. As the true form of the target density matrix is unknown, the choice of operators $\{ H_i \}$ in the Hamiltonian is an educated guess. It is possible that, once the Hamiltonian ansatz is chosen, the space of QBM models does not contain the target, i.e., $\rho_\theta \neq \eta, \forall \theta$. This is called a \emph{model mismatch}, and it is often unavoidable in machine learning. In particular, since we require $m$ to be polynomial in $n$, $\rho_\theta$ cannot encode an arbitrary density matrix. Note that for fixed parameters, $\rho_\theta$ is also known as Gibbs state in the literature. Thus, we refer to an evaluation of the model as the preparation of a Gibbs state.

A natural measure to quantify how well the QBM $\rho_\theta$ approximates the target $\eta$ is the quantum relative entropy
\begin{equation}
\label{eq:RelEnt}
    S(\eta \| \rho_\theta) = \Tr[\eta \log{\eta}] - \Tr[\eta \log{\rho_\theta}]. 
\end{equation} 
This measure generalizes the classical Kullback-Leibler divergence to density matrices.
The relative entropy is exactly zero when the two densities are equal, $\eta=\rho_\theta$, and $S > 0$ otherwise. In addition, when $S(\eta \| \rho_\theta)\leq\epsilon$, by Pinsker's inequality, all Pauli expectation values are within $\mathcal{O}(\sqrt{\epsilon})$. However, achieving this is generally not possible due to the model mismatch. In theory, one can minimize the relative entropy in order to find the optimal model parameters $\theta^{\mathrm{opt}}=\mathrm{argmin}_{\theta}S(\eta \| \rho_\theta)$. The form of the partial derivatives can be obtained analytically and reads
\begin{equation}
\label{eq:RelEntropDeriv}
    \frac{\partial S(\eta\|\rho_\theta)}{\partial \theta_i} = \expval{H_i}_{\rho_\theta} - \expval{H_i}_\eta . 
\end{equation}
This is the difference between the target and model expectation values of the operators that we choose in the ansatz. A stationary point of the relative entropy is obtained when $\expval{H_i}_{\rho_\theta} = \expval{H_i}_\eta$ for $i \in \{1, \dots, m\}$. Since $S$ is strictly convex (see Supplementary Note 2) this stationary point is the unique global minimum. 

It is difficult in practice to quantify how well the QBM is trained by means of relative entropy. In order to accurately estimate $S(\eta\|\rho_\theta)$, one requires access to the entropy of the target and the partition function of the model. Moreover, due to the model mismatch, the optimal QBM may have $S(\eta \| \rho_{\theta^{\mathrm{opt}}})>0$, which is unknown beforehand. In this work, we provide an operational definition of QBM learning based on the size of the gradient $\nabla S(\eta \|\rho_\theta)$.

\begin{definition}[QBM learning problem]
\label{def:problem}

Given a polynomial-space data set $\{s^\mu\}$ obtained from an $n$-qubit target density matrix $\eta$, a target precision $\epsilon>0$, and a fully-visible QBM with Hamiltonian $\mathcal{H}_\theta =\sum_{i=1}^m \theta_i H_i$, find a parameter vector $\theta$ such that with high probability
\begin{equation}
\label{eq:problem_def}
    | \expval{H_i}_{\rho_{\theta}}  - \expval{H_i}_\eta | \leq \epsilon, \qquad \forall i .
\end{equation}
\end{definition}

The expectation values of the target can be obtained from the data set in various ways. For example, for the generative modeling of a classical binary data set one can define a pure quantum state and obtain its expectation values (see Kappen~\cite{Kappen_2020} and Supplementary Note 5).
For the modeling of a target quantum state (density matrix) one can estimate expectation values from the outcomes of measurements performed in different bases.
Due to the polynomial size of the data set, we can only compute properties of the target (and model) to finite precision. For example, suppose that $s^\mu$ are data samples from some unknown probability distribution $P(s)$ and that we are interested in the sample mean. An unbiased estimator for the mean is $\hat{s} = \frac{1}{M}\sum_{\mu=1}^M s^\mu$. The variance of this estimator is $\sigma^2/M$, where $\sigma^2$ is the variance of $P(s)$. By Chebyshev's inequality, with high probability the estimation error is of order $\sigma/\sqrt{M}$. The polynomial size of the data set implies that the error decreases polynomially in general. In light of this, we say the QBM learning problem is solved for any precision $\epsilon > 0$ in Eq.~\eqref{eq:problem_def}. The idea is that the expectation values of the QBM and target should be close enough that one cannot distinguish them without enlarging the data set.

An exact solution to the QBM learning problem always exists by Jaynes' principle~\cite{jaynes1957}: given a set of target expectations $\{ \expval{H_i}_\eta \}$ there exists a Gibbs state $\rho_{\theta^{\mathrm{opt}}} = e^{\sum_i \theta^{\mathrm{opt}}_i H_i}/Z$ such that $| \expval{H_i}_{\rho_{\theta^{\mathrm{opt}}}} - \expval{H_i}_\eta |= 0, \forall i$. We refer to the model corresponding to Jaynes' solution $\theta^{\mathrm{opt}}$ as the optimal model. In QBM learning we try to get as close as possible to this optimal model by minimizing the difference in expectation values. As shown in Supplementary Note 3, any solution to the QBM learning problem implies a bound on the optimal relative entropy, namely 
\begin{equation}
\label{eq:relentropbound}
    S(\eta\|\rho_{\theta}) -  S(\eta\|\rho_{\theta^{\mathrm{opt}}}) \leq 2 \epsilon \| \theta - \theta^\textrm{opt} \|_1 .
\end{equation} 
A similar result can be derived from Proposition A2 in Rouz\'{e} and Stilck Fran\c{c}a~\cite{Rouze_2023}. This means that if one can solve the QBM learning problem to precision $\epsilon \leq \frac{\epsilon^\prime}{2\| \theta - \theta^\textrm{opt} \|_1}$, one can also solve a stronger learning problem based on the relative entropy to precision $\epsilon^\prime$. However, this requires bounding $\| \theta - \theta^\textrm{opt} \|_1$.

We further motivate our problem definition by highlighting important differences with a related problem called quantum Hamiltonian learning~\cite{Wiebe_2014, franca2020, Anshu_2021, Haah_2024, Rouze_2023, Onorati_2023}. This problem can be understood as trying to find a parameter vector $\|\theta-\theta^*\|_1 \leq \epsilon$, where $\theta^*$ is the \emph{true} parameter vector defining a target Gibbs state $\eta$. This is, in fact, a stronger version of the QBM learning problem because finding the true parameters, $\theta^*$, is more demanding than getting $\epsilon$-close to the expectation values of the optimal model from Jaynes' solution. In realistic machine learning settings, we do not know the exact form of $\eta$ a priori and, Hamiltonian learning is in effect equivalent to quantum state tomography. This task has an exponential sample-complexity lower bound~\cite{Aaronson_2017}. Moreover, even if we do know the exact form, minimizing the distance to the true parameters, $\|\theta-\theta^*\|_1$, can make the problem significantly more complicated than QBM learning, and potentially impractical to solve. Consider a single-qubit target $\eta = e^{\mathcal{H}^*} / \Tr[e^{\mathcal{H}^*}]$, with Hamiltonian $\mathcal{H}^*=\theta^* \sigma^z$. We can easily see that near $\expval{\sigma^z}_\eta=1-\epsilon$ the parameter $\theta^*$ diverges. For example, if $\theta^*=2.64$ then $\expval{\sigma^z}_\eta=0.99$ and if $\theta^*=3.8$ then $\expval{\sigma^z}_\eta=0.999$, from which we see that the optimal parameter is very sensitive, $\mathcal{O}(1)$, to small changes, $\mathcal{O}(10^{-2})$, in expectation values. In these cases, Hamiltonian learning is much harder than simply finding a QBM whose expectation value is within $\epsilon$. A similar argument can be made for high-temperature targets near the maximally mixed state. Here the QBM problem is trivial: since all target expectations are $\expval{H_i}_\eta \approx 2^{-n} \Tr[H_i]$, any model sufficiently close to the maximally-mixed state solves the QBM learning problem. In contrast, Hamiltonian learning insists on estimating the unique true parameters.

Let us now discuss a method for solving QBM learning. We approach the problem by iteratively minimizing the quantum relative entropy,  Eq.~\eqref{eq:RelEnt}, by stochastic gradient descent (SGD)~\cite{Garrigos2023}. This method requires access to a stochastic gradient $\hat{g}_{\theta^t}$ computed from a set of samples at time $t$. Recall that our gradient has the form given in Eq.~\eqref{eq:RelEntropDeriv}. The expectations with respect to the target are estimated from a random subset of the data, often called a \emph{mini-batch}. The mini-batch size is a hyper-parameter and determines the precision $\xi$ of each target expectation. Similarly, the QBM model expectations are estimated from measurements, for example using classical shadows~\cite{Aaronson_2017, Huang_2021} of the Gibbs state $\rho_{\theta^t}$ approximately prepared on a quantum device~\cite{Coopmans_2023}. The number of measurements is also a hyper-parameter and determines the precision $\kappa$ of each QBM expectation. We assume that the stochastic gradient is unbiased, i.e., $\mathbb{E}[\hat{g}_{\theta^t}] = \nabla S(\eta \|\rho_\theta)|_{\theta = \theta^t}$, and that each entry of the vector has bounded variance. At iteration $t$, SGD updates the parameters as 
\begin{equation}
\label{eq:sgd_update}
    \theta^{t+1} = \theta^{t} - \gamma^{t} \hat{g}_{\theta^{t}},
\end{equation} 
where $\gamma^t$ is called the learning rate.

\section*{Results}

\subsection*{Sample complexity}

Stochastic gradient descent solves the QBM learning problem with polynomial sample complexity. We state this in the following theorem, which is the main result of our work.

\begin{theorem}[QBM training]
\label{thm:training}
Given a QBM defined by a set of $n$-qubit Pauli operators $\{H_i\}_{i=1}^m$, a precision $\kappa$ for the QBM expectations, a precision $\xi$ for the data expectations, and a target precision $\epsilon$ such that $\kappa^2 + \xi^2 \geq \frac{\epsilon}{2m}$. After 
\begin{equation}
\label{eq:trainingsteps}
    T = \frac{48\delta_0 m^2  (\kappa^2 + \xi^2 )}{\epsilon^4}
\end{equation} 
iterations of stochastic gradient descent on the relative entropy $S(\eta \|\rho_\theta)$ with constant learning rate $\gamma^t=\frac{\epsilon}{4m^2(\kappa^2+\xi^2)}$, we have
\begin{equation}
    \mathrm{min}_{t=1,..,T} \; \mathbb{E}| \expval{H_i}_{\rho_{\theta^t}}  - \expval{H_i}_\eta | \leq \epsilon, \qquad \forall i,
\end{equation} where $\mathbb{E}[\cdot]$ denotes the expectation with respect to the random variable $\theta^t$. 
Each iteration $t \in \{0, \dots, T\}$ requires
\begin{equation}
\label{eq:copies}
    N \in \mathcal{O} \left( \frac{1}{{\kappa^4}} \log \frac{ m} { 1-\lambda^{\frac{1}{T}} } \right)
\end{equation}
preparations of the Gibbs state $\rho_{\theta^t}$, and the success probability of the full algorithm is $\lambda$. Here, $\delta_0 = S(\eta\|\rho_{\theta^0})-S(\eta\|\rho_{\theta^\mathrm{opt}})$ is the relative entropy difference with the optimal model $\rho_{\theta^\mathrm{opt}}$.
\end{theorem}

A detailed proof of this theorem is given in Supplementary Note 3.
It consists of carefully combining three important observations and results. First, we show that the quantum relative entropy for any QBM $\rho_\theta$ is $L$-smooth with $L=2m\max_j \|H_j\|_2^2$, which follows from upper bounding the relative entropy and quantum belief propagation~\cite{Hastings_2007}. A similar result can be found in Rouz\'{e} and Stilck Fran\c{c}a~\cite{Rouze_2023}, Proposition A4. This is then combined with SGD convergence results from the machine learning literature~\cite{Khaled2020, Garrigos2023} to obtain the number of iterations $T$. Finally, we use sampling bounds from quantum shadow tomography to obtain the number of preparations $N$. In this last step, we focus on the shadow tomography protocol in Huang et al.~\cite{Huang_2021}, which restricts our result to Pauli observables $H_i\equiv P_i$, thus $\|H_i\|_2 = 1$. It is possible to extend this to generic two-outcome observables~\cite{Aaronson_2017, Badescu2021} with a polylogarithmic overhead compared to Eq.~\eqref{eq:copies}. Furthermore, for $k$-local Pauli observables, we can improve the result to 
\begin{align}
    N \in \mathcal{O}\left( \frac{3^k}{\kappa^2} \log \frac{ m} { 1-\lambda^{\frac{1}{T}} } \right)
\end{align}
with classical shadows constructed from randomized measurement~\cite{Huang2020} or by using pure thermal shadows~\cite{Coopmans_2023}.

By combining Eqs.~\eqref{eq:trainingsteps} and~\eqref{eq:copies}, we see that the final number of Gibbs state preparations $N_{\mathrm{tot}}\geq T\times N$ scales polynomially with $m$, the number of terms in the QBM Hamiltonian. By our classical memory assumption, we can only have $m \in \mathcal{O}(\mathrm{poly}(n))$. This means that the number of required measurements to solve QBM learning scales polynomially with the number of qubits. Consequently, there are no barren plateaus in the optimization landscape for this problem, resolving an open question from the literature~\cite{Anschuetz_2022, Kieferova_2021} for QBMs without hidden units. Furthermore, also the stronger problem in Eq.~\eqref{eq:relentropbound} can be solved without encountering an exponentially vanishing gradient as we show in Supplementary Note 3. 

We remark, however, that our result does not imply a solution to the stricter Hamiltonian learning problem with polynomially many samples. In order to achieve this one needs some stronger conditions. For example, using $\alpha$-strong convexity for the specific case of geometrically local models, as shown in Anshu et al.~\cite{Anshu_2021}. 
In Supplementary Note 2, we investigate generalizing this by looking at generic Hamiltonians made of low-weight Pauli operators. We perform numerical simulations using differentiable programming~\cite{Baydin2017} and find evidence indicating that even in this case the quantum relative entropy is $\alpha$-strong convex. This means we could potentially achieve an improved sample complexity for QBM learning. We prove the following theorem in Supplementary Note 3.

\begin{theorem}[$\alpha$-strongly convex QBM training]
\label{thm:training_alpha}
Given a QBM defined by a Hamiltonian ansatz $\mathcal{H}_\theta$ such that $S(\eta\|\rho_\theta)$ is $\alpha$-strongly convex, a precision $\kappa$ for the QBM expectations, a precision $\xi$ for the data expectations, and a target precision $\epsilon$ such that $\kappa^2 + \xi^2 \geq \frac{\epsilon}{2m}$. After
\begin{equation}
    T =\frac{18 m^2 (\kappa^2 + \xi^2)}{\alpha^2\epsilon^2}
\end{equation}
iterations of stochastic gradient descent on the relative entropy $S(\eta \|\rho_\theta)$ with learning rate $\gamma^t \leq \frac{1}{4m^2}$, we have
\begin{equation}
    \mathrm{min}_{t=1,..,T} \; \mathbb{E}| \expval{H_i}_{\rho_{\theta^t}}  - \expval{H_i}_\eta | \leq \epsilon, \qquad \forall i.
\end{equation} 
Each iteration requires the number of samples given in Eq.~\eqref{eq:copies}. 
\end{theorem}

Finally, we observe that the sample bound in Theorem~\ref{thm:training} depends on $\delta_0$, the relative entropy difference of the initial and optimal QBMs. This means that if we can lower the initial relative entropy with an educated guess, we also tighten the bound on the QBM learning sample complexity. In this respect, we prove that we can reduce $\delta_0$ by pre-training a subset of the parameters in the Hamiltonian ansatz. Thus, pre-training reduces the number of iterations required to reach the global minimum.

\begin{theorem}[QBM pre-training]
\label{thm:pretraining}
    Assume a target $\eta$ and a QBM model $\rho_\theta=e^{\sum_{i=1}^m \theta_i H_i}/Z$ for which we like to minimize the relative entropy $S(\eta \|\rho_\theta)$. Initializing $\theta^0=0$ and pre-training $S(\eta\| \rho_\theta)$ with respect to any subset of $\tilde{m} \leq m $ parameters guarantees that \begin{equation}\label{eq:pretrainrel}
         S(\eta \|\rho_{\theta^\mathrm{pre}}) \leq S(\eta \|\rho_{\theta^0}), 
    \end{equation} where $\theta^\mathrm{pre} = [\chi^{\mathrm{pre}}, 0_{m-\tilde{m}}]$ and the vector $\chi^{\mathrm{pre}}$ of length $\tilde{m}$ contains the parameters for the terms $\{H_i\}_{i=1}^{\tilde{m}}$ at the end of pre-training.
    More precisely, starting from $\rho_\chi=e^{\sum_{i=1}^{\tilde{m}} \chi_i H_i}/Z$ and minimizing $S(\eta \|\rho_\chi)$ with respect to $\chi$ ensures Eq.~\eqref{eq:pretrainrel} for any $S(\eta \|\rho_{\chi^{\mathrm{pre}}})\leq S(\eta \|\rho_{\chi^0})$.
\end{theorem}

We refer to Supplementary Note 4 for the proof of this theorem. It applies to any method that is able to minimize the relative entropy with respect to a subset of the parameters. Crucially, all the other parameters are fixed to zero, and the pre-training needs to start from the maximally mixed state $\mathbb{I}/2^n$. Note that the maximally mixed state is not the most `distant' state from the target state, and there exist states for which $S(\eta \| \rho_\theta) > S(\eta \| \mathbb{I}/2^n)$. As an example of pre-training, one could use SGD as described above, and apply updates only to the chosen subset of parameters. With a suitable learning rate, this ensures that pre-training lowers the relative entropy compared to the maximally mixed state $S(\eta \|\mathbb{I}/2^n)$. As a consequence, one can always add additional linearly independent terms to the QBM ansatz without having to retrain the model from scratch. The performance is guaranteed to improve, specifically towards the global optimum due to the strict convexity of the relative entropy. By Eq.~\eqref{eq:pretrainrel} and Pinsker's inequality, the trace-distance to all the other observables reduces as well. This is in contrast to other QML models which do not have a convex loss function. It is particularly useful if one can pre-train a certain subset classically before training the full model on a quantum device. For instance, in Supplementary Note 4 we present mean-field and Gaussian Fermionic pre-training strategies with closed-form expressions for the optimal subset of parameters.

\begin{figure*}[t]
    \centering    
    \includegraphics[width=.95\linewidth]{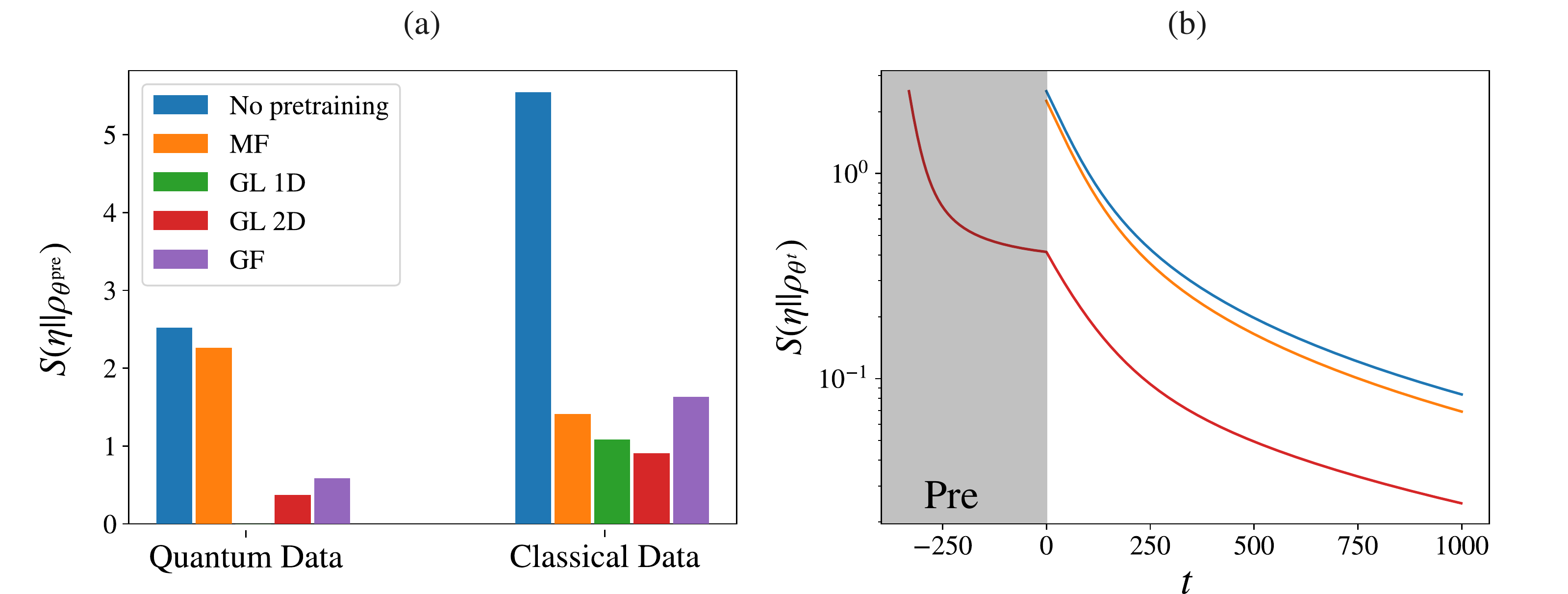}
    \caption{Pre-training and training quantum Boltzmann machines. (a) Quantum relative entropy $S(\eta\lvert\rho_{\theta^\mathrm{pre}})$ obtained after various pre-training strategies. We compare a mean-field (MF) model, a one-dimensional and two-dimensional geometrically local (GL) model, and a Gaussian Fermionic (GF) model to no pre-training/maximally mixed state. For the GL models, we stop the pre-training after the pre-training gradient is smaller than $0.01$. We consider an $8$-qubit target $\eta$ as the Gibbs state $e^{\mathcal{H}_{\mathrm{XXZ}}}/Z$ of a one-dimensional XXZ model (Quantum Data), and a target $\eta$ which coherently encodes the binary salamander retina data set (Classical Data). (b) Quantum relative entropy versus number of iterations for Quantum Data. The $t < 0$ iterations (gray area) show the reduction in relative entropy for GL 2D pre-training (red line). The $t=0$ iteration corresponds to the pre-training results in panel (a). The $t>0$ iterations show the training results in the absence of gradient noise, i.e., $\kappa=\xi=0$.}
    \label{fig:noiseless_results}
\end{figure*}

\subsection*{Numerical experiments}

In order to verify our theoretical findings we perform numerical experiments of QBM learning on data sets constructed from a quantum and a classical source. In particular, we use expectation values of the Gibbs state of a 1D quantum XXZ model in an external field~\cite{Heisenberg1928, Franchini_2017}, and expectation values of a classical salamander retina data set~\cite{Tkacik2014}. How to obtain the target state expectation values for both cases is explained in Supplementary Note 5.

First we focus on reducing the initial relative entropy $S(\eta \|\rho_{\theta^0})$ by QBM pre-training, following Theorem~\ref{thm:pretraining}. We consider mean-field (MF), Gaussian Fermionic (GF), and geometrically local (GL) models as the pre-training strategies. The Hamiltonian ansatz of the MF model consists of all possible one-qubit Pauli terms $\{H_i\}_{i=1}^{3n}=\{ \sigma_i^x, \sigma_i^y, \sigma_i^z\}_{i=1}^{n}$ in Eq.~\eqref{eq:generic_H}, hence it has $3n$ parameters. The Hamiltonian of the GF model is a quadratic form $\mathcal{H}_\theta^{\mathrm{GF}}= \sum_{i,j}\tilde{\theta}_{ij}\vec{
C}_i^\dagger \vec{C}_j$ of Fermionic creation and annihilation operators, where $\tilde{\theta}_{ij}$ is the $2n\times 2n$ Hermitian parameter matrix, which has $n^2$ free parameters. Here $\vec{C}^\dagger = [c_1^\dagger, \dots, c_n^\dagger, c_1, \dots, c_n]$ with the operators satisfying $\{c_i, c_j\} = 0$ and $\{c_i^\dagger, c_j\}=\delta_{ij}$, where $\{A, B\} = AB + BA$ is the anti-commutator. The advantage of the MF and GF pre-training is that there exists a closed-form solution given target expectation values $\langle H_i \rangle_\eta$ (see Supplementary Note 4).

In contrast, the GL models are defined with a Hamiltonian ansatz

\begin{equation} 
\label{eq:glHam}
    \mathcal{H}_\theta^{\mathrm{GL}} = \sum_{k=x,y,z}\sum_{\langle i, j\rangle}\lambda^k_{ij}\sigma_i^k\sigma^k_j + \sum_i^n \gamma^k_i\sigma^k_i,
\end{equation} 
for which, in general, the parameter vector $\vec{\theta}\equiv\{\lambda, \sigma\}$ cannot be found analytically. Here the sum $\langle i, j\rangle$ imposes some constraints on the (geometric) locality of the model, i.e., summing over all possible nearest neighbors in some $d$-dimensional lattice. In particular, we choose one- and two-dimensional locality constraints suitable with the assumptions given in previous work~\cite{Anshu_2021, Haah_2024}. In these specific cases the relative entropy is strongly convex, thus pre-training has the improved performance guarantees from Theorem~\ref{thm:training_alpha}.

In Fig.~\ref{fig:noiseless_results}~(a) we show the  relative entropy $S(\eta\|\rho_{\theta^{\mathrm{pre}}})$ obtained after pre-training with these models for $8$-qubit problems. For MF and GF we use the closed form solutions, and for GL we use exact gradient descent with learning rate of $\gamma=1/\tilde{m}$. As a comparison, we also show the result obtained for the QBM without pre-training, i.e., for the maximally mixed state $S(\eta\|\rho_{\theta=0})$.
We observe that for both targets, all pre-training strategies are successful in obtaining a lower $S(\eta\|\rho_{\theta^{\mathrm{pre}}})$ compared to the maximally mixed state. For the GL 1D ansatz, the target quantum state is contained within the QBM model space, which means that the relative entropy is zero after pre-training. This shows that having knowledge about the target (e.g., the fact that it is one dimensional) could help inform QBM ansatz design and significantly reduce the complexity of QBM learning. The GF model, which has completely different terms in the ansatz, manages to reduce $S(\eta\|\rho_{\theta^{\mathrm{pre}}})$ by a factor of $\approx 5$ for the quantum target and $\approx 4$ for the classical target. By the Jordan-Wigner transformation, we can express the target XXZ model in a Fermionic basis. In this representation, the target only has a small perturbation compared to the model space of the GF model. This could explain the good performance of pre-training with the GF model. 

Next, we investigate the effect of using pre-trained models as starting point for QBM learning with exact gradient descent. We extend all pre-trained models with additional linearly independent terms in the ansatz
\begin{equation} 
\label{eq:fully_connect_qbm}
    \mathcal{H}_\theta = \sum_{k=x,y,z}\sum_{i, j>i}\lambda^k_{ij}\sigma_i^k\sigma^k_j + \sum_i^n \gamma^k_i\sigma^k_i,
\end{equation} 
where now, compared to Eq.~\eqref{eq:glHam}, any qubit can be connected to any other qubit, and we do not have a constraint on the geometric locality. This is the fully-connected QBM Hamiltonian used in previous work~\cite{Kappen_2020, Coopmans_2023}. We now consider data from the quantum target $\eta = e^{{\mathcal{H}_\mathrm{XXZ}}}/Z$ for $8$ qubits.

In Fig.~\ref{fig:noiseless_results}~(b) we show the quantum relative entropy during training when starting from different pre-trained models $\rho_{\theta^0} \coloneqq \rho_{\theta^\mathrm{pre}}$. The initial parameter vector $\theta^0$ is the one obtained at the end of pre-training on quantum data in panel (a). We do not include gradient noise in this simulation, i.e., $\kappa=\xi=0$, and we use a learning rate of $\gamma=1/(2m)$. The performance of the MF pre-trained model (orange line) is better than that of the vanilla model (blue line) for all $t$. However, the improvement is modest. Starting from a GL 2D  model (red line) has a much more significant effect, with $S(\eta \|\rho_{\theta^t})$ being an order of magnitude smaller than that of the vanilla model at all iterations $t$. This pre-training strategy requires very few gradient descent iterations (dark gray area). This could be stemming from the strong convexity of this particular pre-training model. In general, the benefits of pre-training needs to be assessed on a case-by-case basis as the size of the improvement depends on the particular target and the particular pre-training model used. Importantly, we only proved Theorem~\ref{thm:training} for a learning rate of $\gamma = \min \{ \frac{1}{L}, \frac{\epsilon}{4m^2(\kappa^2 +\xi^2)}\}$. Therefore, choosing a larger learning rate might reduce the benefits of pre-training. 

\begin{figure*}[t]
    \centering
    \includegraphics[width=.5\linewidth]{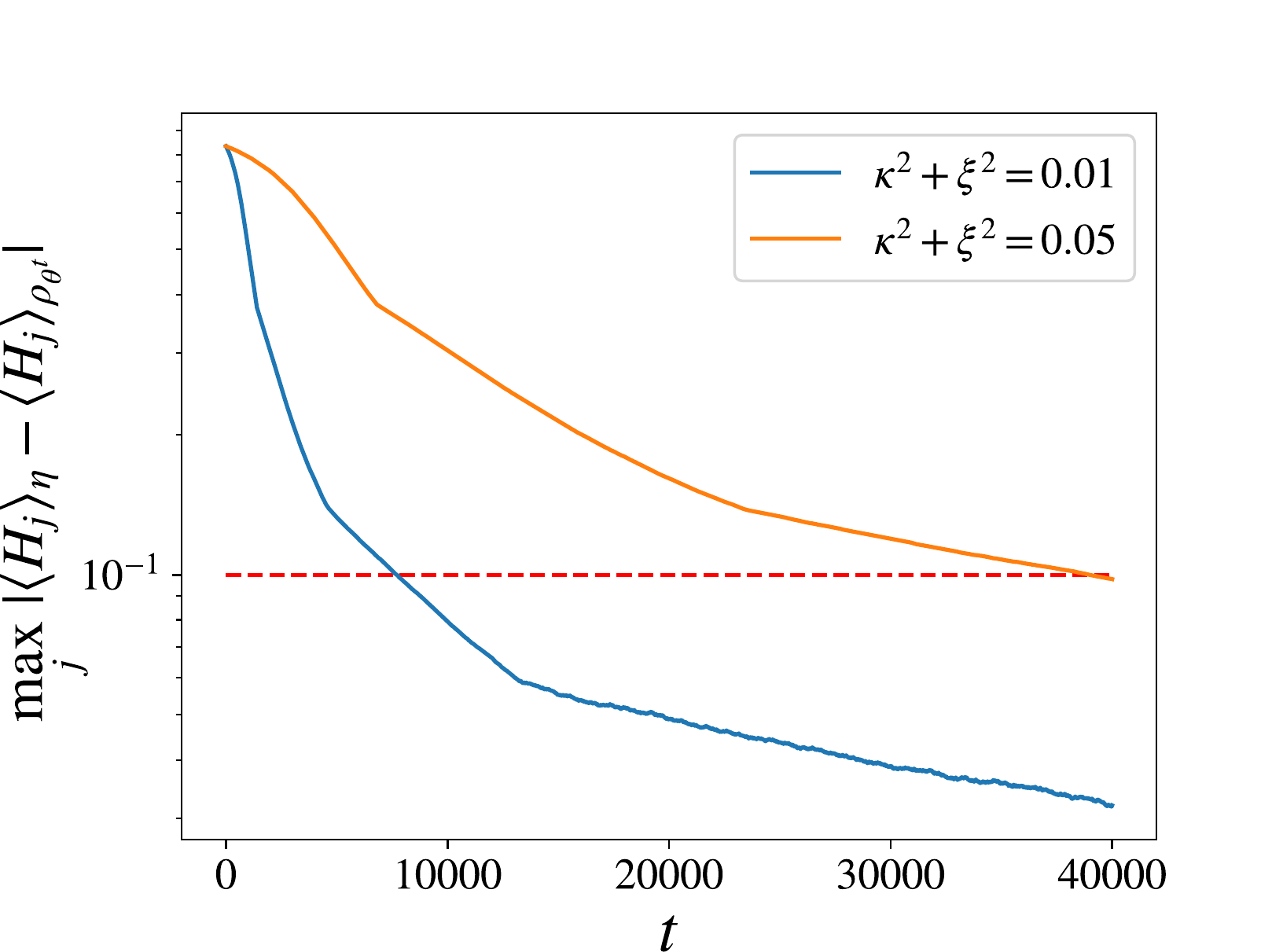}
    \caption{The maximum error in the expectation values versus the number of iterations of stochastic gradient descent (SGD). The target $\eta$ is made of $8$ features from the classical salamander retina data set, the model $\rho_{\theta^t}$ is a quantum Boltzmann machine (QBM). When computing expectation values for the gradient, we allow precision $\xi$ for the target and precision $\kappa$ for the QBM. We compare the combined noise strength of $0.01$ (blue line) to $0.05$ (orange line). We aim for a maximum error of $\epsilon=0.1$ (red dashed line) and use a learning rate of $\gamma = \frac{\epsilon}{4m^2(\kappa^2+\xi^2)}$. SGD converges within a number of iterations consistent with Theorem~\ref{thm:training}.}
    \label{fig:noisy_results}
\end{figure*}

Lastly, we numerically confirm our bound on the number of SGD iterations (Eq.~\eqref{eq:trainingsteps} in Theorem~\ref{thm:training}). We consider data from the classical salamander retina target with $8$ features and a fully-connected QBM model on $8$ qubits. 
In Fig.~\ref{fig:noisy_results} we compare training with two different noise strengths $\kappa^2+\xi^2$. We implemented these settings by adding Gaussian noise effectively simulating the noise strength determined by the number of data points in the data set and the number of measurements of the Gibbs state on a quantum device. Using the standard Monte Carlo estimate, at each update iteration we require a mini-batch of data samples of size $~\frac{1}{\xi^2}$, and a number of measurements $\frac{1}{\kappa^2}$ (assuming these measurements can be performed without additional hardware noise). We could even use mini batches of size $1$ and a single measurement, as long as the expectation values are unbiased. For both noise strengths, we obtain the desired target precision of $\epsilon=0.1$ within 40000 iterations. This is well within the bound of order $10^9$ on the number of iterations obtained from Theorem~\ref{thm:training}, which is the worst-case scenario. For the interested reader, in Supplementary Note 3 we provide additional numerical evidence confirming the scaling behavior of our theorems.

\subsection*{Generalization}

Before concluding our work, in this section we briefly comment on the generalization capability of the fully-connected QBM to which our convergence results apply. In particular, we connect with a result proved by Aaronson~\cite{Aaronson_2007}, which concerns the learnability of quantum states.

\begin{theorem}[Theorem 1.1 in ~\cite{Aaronson_2007}]
\label{thm:occams}
Let $\eta$ be a $2^n$-dimensional mixed state, and let $\mathcal{D}$ be any probability measure over two-outcome measurements. Then given samples
$E_1, \dots , E_m$ drawn independently from $\mathcal{D}$, with probability at least $1 - \nu$, the samples have the following generalization property: any hypothesis state $\rho$ such that $| \expval{E_i}_\rho - \expval{E_i}_\eta | \leq \zeta\varepsilon$ for all $i \in \{1, \dots, m\}$ will also satisfy
\begin{equation}
        \mathrm{Pr}_{E\sim \mathcal{D}} \left( |\expval{E}_\rho - \expval{E}_\eta |\leq \varepsilon\right) \geq 1 - \zeta,
\end{equation}
provided we took
\begin{equation}
\label{eq:m_ops}
    m\geq \frac{C}{\zeta^2 \varepsilon^2}\left(\frac{n}{\zeta^2 \varepsilon^2}\log^2 \frac{1}{\zeta\varepsilon} + \log \frac{1}{\nu} \right)
\end{equation}
for some large enough constant $C$. 
\end{theorem}

In other words, when a hypothesis state $\rho$ matches the expectation values of $m$ randomly sampled measurement operators from distribution $\mathcal{D}$, then, with high probability, most expectation values of the other measurement operators from $\mathcal{D}$ match up to error $\varepsilon$. Comparing to Definition~\ref{def:problem}, we can identify the learned QBM $\rho_\theta$ as a specific hypothesis state that matches the expectation values of operators $\{H_i\}_{i=1}^m$ of the target $\eta$. Our convergence results show how many Gibbs state samples are required to obtain such a hypothesis state. Thus, our QBM learning result in Theorem~\ref{thm:training}, in a sense, complements Theorem~\ref{thm:occams} and provides a concrete algorithm for finding the hypothesis $\rho$ with $\mathcal{O}(\mathrm{poly}(n))$ samples. One distinction is that in QBM learning we already assume access to a fixed set of measurement operators which defines the Hamiltonian ansatz, instead of randomly sampling operators from a distribution. 

In order to say something about generalization, we now consider a scenario in which the user is interested in solving the QBM learning problem for some fixed (potentially exponentially large) number of two-outcome measurement operators $\{H_i\}_{i=1}^{K}$. Then by defining $\mathcal{D}$ as the uniform distribution over the set $\{H_i\}_{i=1}^{K}$, and constructing the QBM ansatz by sampling $m < K$ operators from $\mathcal{D}$, we can directly apply Theorem~\ref{thm:occams} for generalization. In particular, defining the ansatz in this way ensures that our trained QBM will, with high probability, generalize to any other observable sampled from $\mathcal{D}$, including ones that are not in the ansatz.

\section*{Conclusions}

In this paper, we give an operational definition of quantum Boltzmann machine (QBM) learning. We show that this problem can be solved with polynomially many preparations of quantum Gibbs states. To prove our bounds we use the properties of the quantum relative entropy in combination with the performance guarantees of stochastic gradient descent (SGD). We do not make any assumption on the form of the Hamiltonian ansatz, other than that it consists of polynomially many terms. This is in contrast with earlier works that looked at the related Hamiltonian learning problem for geometrically local models~\cite{Anshu_2021, Haah_2024}. There, strong convexity is required in order to relate the optimal Hamiltonian parameters to the Gibbs state expectation values. In our machine learning setting, we do not know the form of the target Hamiltonian a priori. Therefore, we argue that learning the exact parameters is irrelevant, and one should focus directly on the expectation values. Our bounds only require $L$-smoothness of the relative entropy and apply to all types of QBMs without hidden units.

We also show that our theoretical sampling bounds can be tightened by lowering the initial relative entropy of the learning process. We prove that pre-training on any subset of the parameters is guaranteed to perform better than (or equal to) the maximally mixed state. This is beneficial if one can efficiently perform the pre-training, which we show is possible for mean-field, Gaussian Fermionic, and geometrically local QBMs. We verify the performance of these models and our theoretical bounds with classical numerical simulations. From this, we learn that knowledge about the target (e.g., its dimension, degrees of freedom, etc.) can significantly improve the training process. Furthermore, we find that our generic bounds are quite loose, and in practice, one could get away with a much smaller number of samples.  

This brings us to interesting avenues for future work. One possibility consists of tightening the sample bound by going beyond the plain SGD method that we have used here. This could be done by adding momentum, by using other advanced update schemes~\cite{foster2019complexity, Garrigos2023}, or by exploiting the convexity of the relative entropy. While we believe this can improve the $\mathcal{O}(\mathrm{poly}(m, \frac{1}{\epsilon}))$ scaling in our bounds, we note that it does not change the main conclusion of our paper: the QBM learning problem can be solved with polynomially many preparations of Gibbs states.
Another important direction is the investigation of different ans\"{a}tze and their performance. Generative models are often assessed in terms of training quality~\cite{riofrio2023performance}, but generalization capabilities have been recently investigated by both classical~\cite{zhao2018bias,thanhtung2021generalization} and quantum~\cite{du2022power,gili2023generalization} machine learning researchers. For the case of QBMs, we employ an alternative definition of generalization~\cite{Aaronson_2007} and show how it can be used to construct ans\"{a}tze. The numerical investigation of this method is an interesting venue for future work.

Our pre-training result could be useful for implementing QBM learning on near-term and early fault-tolerant quantum devices. To this end, one would use a quantum computer as a Gibbs sampler. There exists a plethora of quantum algorithms (e.g., see~\cite{Temme_2011, Chowdhury_2016, Holmes_2022, chifang2023quantum, zhang2023dissipative}) that prepare Gibbs states with a quadratic improvement in time complexity over the best existing classical algorithms. Moreover, the use of a quantum device gives an exponential reduction in space complexity in general. One can also sidestep the Gibbs state preparation and use algorithms that directly estimate Gibbs-state expectation values, e.g., by constructing classical shadows of pure thermal quantum states~\cite{Coopmans_2023}. This reduces the number of qubits and, potentially, the circuit depth.

Finally, our results open the door to novel methods for the incremental learning of QBMs driven by the availability of both training data and quantum hardware. For example, one could select a Hamiltonian ansatz that is very well suited for a particular quantum device. After exhausting all available classical resources during the pre-training, one enlarges the model and continues the training on the quantum device, which is guaranteed to improve the performance. As the quantum hardware matures, it allows us to execute deeper circuits and to further increase the model size. Incremental QBM training strategies could be designed to follow the quantum hardware roadmap, towards training ever larger and more expressive quantum machine learning models.

\newpage

\section*{Data availability}
The data that support the findings of this study are available at Zenodo~\cite{zenodo_files}.

\section*{Acknowledgments}
We thank Guillaume Garrigos and Robert M. Gower for pointing us to the literature on stochastic gradient descent. We thank Samuel Duffield, Yuta Kikuchi, Mark Koch, Enrico Rinaldi, Matthias Rosenkranz, and Oscar Watts for helpful discussions and for their feedback on an earlier version of this manuscript. We also thank Dhrumil Patel and Mark M. Wilde for identifying an error in the previous version of this manuscript.

\section*{Author contributions}
L.C. and M.B. contributed equally to this work.

\section*{Competing interests}
L.C and M.B. are affiliated with Quantinuum Ltd and the results presented in this paper are subject to a pending patent application filed by Quantinuum Ltd with Application Number PCT/GB2024/051593.

\let\oldaddcontentsline\addcontentsline
\renewcommand{\addcontentsline}[3]{}
\bibliography{refs}
\let\addcontentsline\oldaddcontentsline

\setcounter{figure}{0}
\makeatletter
\renewcommand{\fnum@figure}{\textbf{Supplementary Figure \thefigure}}
\makeatother

\clearpage
\appendix
\onecolumngrid
\begin{center}
	\noindent\textbf{\large{Supplementary Information for\\``On the Sample Complexity of Quantum Boltzmann Machine Learning''}}\\
    \vspace{.5cm}
    Luuk Coopmans and Marcello Benedetti\\
    \small{\textit{Quantinuum, Partnership House, Carlisle Place, London SW1P 1BX, United Kingdom}}\\
    \small{(Dated: August 22, 2024)}
\end{center}

\renewcommand{\baselinestretch}{1.1}\normalsize
\tableofcontents
\renewcommand{\baselinestretch}{1.0}\normalsize

\clearpage

\section*{\MakeUppercase{Supplementary Note 1: Some useful mathematical facts and relations}}
\addcontentsline{toc}{section}{Supplementary Note 1: Some useful mathematical facts and relations}

Here we state and derive some useful mathematical results that are used in the proofs in the other appendices. The definitions and lemmas about convexity are taken from~\cite{Garrigos2023}. The derivative of the matrix exponential expressed as quantum belief propagation is a result of~\cite{Hastings_2007} and we rederive it in a pedagogical way.

\subsection*{Convexity}
\addcontentsline{toc}{subsection}{Convexity}

\begin{definition}[Convexity]
\label{def:convexity}
A multivariate function $f:\mathbb{R}^m \mapsto \mathbb{R}$ is said to be convex when 
\begin{equation}
    f(tx + (1-t)y) \leq tf(x) + (1-t)f(y), \quad \forall x, y\in\mathbb{R}^m, \quad t\in [0, 1]. 
\end{equation}
If additionally the gradient $\nabla f(x^*)$ is zero only for one unique vector $x^* \in\mathbb{R}^m$, then $f$ is said to be strictly convex.
\end{definition}

The following Lemma can be deduced from the standard definition of convexity~\cite{Garrigos2023}. 

\begin{lemma}
\label{lem:convexity}
Let $f$ be twice continuously differentiable. Then $f$ is convex if 
\begin{equation}
\label{eq:convex} 
    v^T \nabla^2 f(x) v \geq 0, \quad \forall x,v \in \mathbb{R}^m . 
\end{equation}
\end{lemma} 

A stronger version of convexity is used in some of our discussions.

\begin{definition}[$\alpha$-Polyak-{\L}ojasiewicz]
\label{def:PL}
Let $f : \mathbb{R}^m \rightarrow \mathbb{R}$, and $\alpha > 0$. We say that $f$ is $\alpha$-Polyak-{\L}ojasiewicz if 
\begin{equation}
\frac{1}{2\alpha}\|\nabla f(x)\|^2 \geq f(x) - \mathrm{min}_x f(x),
\end{equation} 
where $\| \cdot \|$ is the Euclidean norm.
\end{definition}

An even stronger convexity condition is the following.

\begin{definition}[$\alpha$-strong convexity]
Let $f : \mathbb{R}^m \rightarrow \mathbb{R}$, and $\alpha > 0$. We say that $f$ is $\alpha$-strongly convex if
\begin{equation}
    \frac{ \alpha t(1-t) } {2} \| x - y \| ^2 
+ f(tx + (1-t)y) \leq tf (x) + (1-t)f(y).
\end{equation} 
\end{definition}

The latter implies the former.

\begin{lemma}
\label{lem:PL}
If $f$ is $\alpha$-strongly convex then $f$ is $\alpha$-Polyak-{\L}ojasiewicz.
\end{lemma}

The strong convexity of a function can be tested as follows.

\begin{lemma}
\label{lem:strong_convexity}
Let $f$ be twice continuously differentiable. Then $f$ is $\alpha$-strongly convex if 
\begin{equation}
    v^T \nabla^2 f(x) v \geq \alpha \| v \|^2, \quad \forall x,v \in \mathbb{R}^m .
\end{equation}
\end{lemma}

Besides convexity, we also need to characterize the smoothness of a function.

\begin{definition}[$L$-smoothness]
\label{def:smooth}
Let $f:\mathbb{R}^m \rightarrow \mathbb{R}$ and $L > 0$. We say that f is $L$-smooth if it is differentiable
and if the gradient $\nabla f$ is $L$-Lipschitz:
\begin{equation}
\| \nabla f(x) - \nabla f(y) \| \leq L \| x - y \|, \quad \forall x, y \in \mathbb{R}^m.
\end{equation}
\end{definition}

For $L$-smooth functions we have the following useful property~\cite{Garrigos2023}.

\begin{lemma}[Descent lemma]
\label{lem:descent}
Let $f: \mathbb{R}^m \rightarrow \mathbb{R}$ be a twice differentiable, $L$–smooth function, then
\begin{align}
f(y) \leq f(x) + \nabla f(x)^T (y-x) + \frac{L}{2} \| y -x \|^2 .
\end{align}
\end{lemma}

\subsection*{Derivative of a matrix exponential}
\addcontentsline{toc}{subsection}{Derivative of a matrix exponential}

\begin{lemma}[Quantum belief propagation~\cite{Hastings_2007}] 
Let $H=W+\theta V$ be a Hamiltonian ansatz with parameter $\theta$. Let $f(t)$ be the function for which $\int_{-\infty}^\infty f(t) e^{-i t \omega} dt = \frac{\tanh(\omega/2)}{\omega/2}$, and define the operator Fourier transform as $\Phi( \cdot ) = \int_{-\infty}^{\infty} f(t) e^{itH} \cdot e^{-itH} dt$. Then
\begin{align}
\partial_\theta e^H = \frac{1}{2} \left\{ \Phi(V) , e^H\right\},
\end{align}
where $\{A, B\} = AB + BA$ is the anti-commutator.
\end{lemma}

\begin{proof}
The derivative of the matrix exponential $e^H$ with respect to a parameter is given by Duhamel's formula~\cite{Haber_2018}
\begin{align}
\label{eq:duhamel}
    \partial_\theta e^H &= \int_0^1 e^{(1-s)H} (\partial_\theta H) e^{s H} ds .
\end{align}
Taking $H=W+\theta V$, with simple manipulations we find a useful alternative expression
\begin{align}
\label{eq:manip1}
\begin{split}
    \partial_\theta e^H &= e^H \int_0^1 e^{-sH} V e^{s H} ds \\
    &= \sum_{j,k} \ketbra{j}{k} \matrixel{j}{V}{k} e^{\lambda_j} \int_0^1 e^{s(\lambda_k - \lambda_j)} ds \\
    &= \sum_{j,k} \ketbra{j}{k} V_{jk} e^{\lambda_j} \; \frac{e^{\lambda_k - \lambda_j} - 1}{\lambda_k - \lambda_j} .
\end{split}
\end{align}
Here we use the basis diagonalizing the Hamiltonian, $H = \sum_{j} \lambda_j \dyad{j}$, and we introduce the notation $V_{jk} = \matrixel{j}{V}{k}$. The above expression is valid also for the diagonal entries, $k=j$, since $\lim_{x \rightarrow 0} \frac{e^x - 1}{x} = 1$.  Now,
\begin{align}
    e^{\lambda_j} \; \frac{e^{\lambda_k - \lambda_j} - 1}{\lambda_k - \lambda_j}  =  e^{\lambda_j} \; \frac{e^{\lambda_k - \lambda_j} - 1}{e^{\lambda_k - \lambda_j} +1 } \; \frac{e^{\lambda_k - \lambda_j} +1 }{\lambda_k - \lambda_j}  =  \frac{\tanh( \tfrac{\lambda_k - \lambda_j}{2})}{ \tfrac{\lambda_k - \lambda_j}{2} } \; \frac{e^{\lambda_k} + e^{\lambda_j} }{2} .
\end{align}
With the notation $\hat{f}(\omega) = \frac{\tanh(\omega/2)}{\omega/2}$ we can write
\begin{align}
    \partial_\theta e^H &= \sum_{j,k} \ketbra{j}{k} V_{jk} \hat{f}(\lambda_k - \lambda_j) \; \frac{e^{\lambda_k} + e^{\lambda_j} }{2} .
\end{align}
Let us interpret $\hat{f}(\omega)$ as the Fourier transform of another function: $\hat{f}(\omega) = \int_{-\infty}^\infty f(t) e^{-i t \omega} dt$. Plugging this in the previous expression we obtain
\begin{align}
\begin{split}
\label{eq:ExpDerCommutator} 
    \partial_\theta e^H &= \sum_{j,k} \ketbra{j}{k} V_{jk} \int_{-\infty}^{\infty} f(t) e^{-it(\lambda_k - \lambda_j)} dt \left( \frac{e^{\lambda_k}}{2} + \frac{e^{\lambda_j}}{2} \right) \\
    &= \frac{1}{2} \int_{-\infty}^{\infty} f(t) \; \sum_{j} e^{it\lambda_j} \dyad{j} \; V \; \sum_{k} e^{-it\lambda_k + \lambda_k} \dyad{k} dt + \frac{1}{2}\int_{-\infty}^{\infty} f(t) \; \sum_{j} e^{it\lambda_j + \lambda_j} \dyad{j} \; V \; \sum_{k} e^{-it\lambda_k} \dyad{k} dt \\
    &= \frac{1}{2} \int_{-\infty}^{\infty} f(t) e^{itH} V e^{-itH} dt e^{H} + \frac{1}{2} e^H \int_{-\infty}^{\infty} f(t) e^{itH} V e^{-itH} dt  \\
    &= \frac{1}{2} \left\{ \Phi(V) , e^H\right\} .
\end{split}
\end{align}
We have recovered a result from~\cite{Hastings_2007} by different means. 
\end{proof}

\clearpage

\section*{\MakeUppercase{Supplementary Note 2: Properties of the quantum relative entropy for QBMs}}
\addcontentsline{toc}{section}{Supplementary Note 2: Properties of the quantum relative entropy for QBMs}

In this appendix, we prove some properties of the quantum relative entropy $S(\eta \|\rho_\theta)$ of a generic Quantum Boltzmann Machine (QBM) $\rho_\theta$ with respect to some arbitrary target $\eta$. These properties are used for the proof of the theorems in the main text. We start by showing the convexity and afterward, we show the $L$-smoothness. 

\subsection*{Strict convexity} 
\addcontentsline{toc}{subsection}{Strict convexity}

\begin{lemma}
Let $\mathcal{H}_\theta = \sum_{i=1}^m \theta_{i} H_{i}$ be a Hamiltonian ansatz where $H_i$ are non-commuting operators, and $\rho_\theta = \frac{e^{\mathcal{H}_\theta}}{Z}$ the corresponding QBM. For all states $\eta$, the quantum relative entropy $S(\eta\|\rho_\theta)$ is a strictly convex function of $\theta$. 
\end{lemma}

\begin{proof}
In order to show strict convexity of $S$ we use Lemma~\ref{lem:convexity}. We first show that the Hessian of the quantum relative entropy with respect to the QBM parameters, $\nabla^2 S$, is positive semidefinite. Afterwards, we show that $S$ has a unique global optimizer $\theta^*$ for which $\nabla S(\eta \|\rho_{\theta^*})=0$, and apply the Lemma. 

Using the derivative of the matrix exponential in Eq.~\eqref{eq:ExpDerCommutator} we have
\begin{align}
\begin{split}
\label{eq:deriv_alt}
    \frac{\partial S}{\partial \theta_{i}} &= \frac{\partial}{\partial \theta_{i}} \Tr[ \eta \left( \log{\eta}-\mathcal{H}_\theta+\log \Tr[e^{\mathcal{H}_\theta} ] \right ) ] \\
    &= -\Tr[\eta H_{i}] + \frac{ \Tr[ \{\Phi(H_i), e^{\mathcal{H}_\theta} \} ] } {2 \Tr[e^{\mathcal{H}_\theta}] }  \\
    &= -\Tr[\eta H_{i}] +\Tr[\rho_\theta \Phi(H_{i})] \\
    &= -\Tr[\eta H_{i}] +\Tr[\rho_\theta H_{i}].
\end{split}
\end{align}
In the last step, we use the cyclic property of the trace. This is Eq.~\eqref{eq:RelEntropDeriv} in the main text. Next, we take the second derivative
\begin{align}
\begin{split}
\label{eq:hessian_new}
    \frac{\partial^2 S}{\partial \theta_{i} \partial \theta_{j}} &=  \frac{\partial}{\partial \theta_{j} } \Tr[\rho_\theta H_{i}] \\
    &= \Tr[\left( \frac{ \{ \Phi(H_j), e^{\mathcal{H}_\theta} \} } {2 \Tr[e^{\mathcal{H}_\theta} ] }  - \frac{ e^{\mathcal{H}_\theta} \Tr[ \{ \Phi(H_j), e^{\mathcal{H}_\theta} \} ] } {2 (\Tr[e^{\mathcal{H}_\theta} ])^2} \right)  H_{i} ]\\
    &= \frac{1}{2} \Tr[ \rho_\theta \{ \Phi(H_j), H_i \}] - \Tr[ \rho_\theta H_j] \Tr[ \rho_\theta H_i] . 
\end{split}
\end{align}
In the last step we used $\Tr[ A \{B, C\}] = \Tr[ C \{A, B\}]$ to rearrange the terms. We now show that the Hessian is positive semidefinite and satisfies Eq.~\eqref{eq:convex}. Note that one can arrive at the same result using Lemma 27 in~\cite{Anshu_2021}). For any vector $v \in \mathbb{R}^m$, we have
\begin{align}
\begin{split}
\label{eq:positive_semidefinite_0}
    v^T \nabla^2 S v &=  \sum_{n,m} v_n v_m \left( \frac{1}{2} \Tr[ \rho_\theta \{ \Phi(H_n), H_m \}] - \Tr[ \rho_\theta H_n] \Tr[ \rho_\theta H_m] \right) \\
    &= \frac{1}{2} \Tr[ \rho_\theta \left\{ \Phi\left(\sum_n v_n H_n\right), \sum_m v_m H_m \right\}] - \Tr[ \rho_\theta \sum_n v_n H_n] \Tr[ \rho_\theta \sum_m v_m H_m] \\
    &= \frac{1}{2} \Tr[ \rho_\theta \{ \Phi(W), W \}] - \Tr[ \rho_\theta W] \Tr[ \rho_\theta W] \\
    &= \Re \Tr[\rho_\theta \Phi(W) W ] - \Tr[ \rho_\theta W]^2 ,
\end{split}
\end{align}
where we define the Hermitian operator $W = \sum_n v_n H_n$ and use that $\Phi(W)$ is also Hermitian. We denote the real part of a complex number by $\Re$. Let us inspect the first term in the basis diagonalizing the Hamiltonian
\begin{align}
\begin{split}
    \Re \Tr[\rho_\theta \Phi(W) W ] &= \Re \Tr[ \left(\sum_{j} \frac{e^{\lambda_j}}{\Tr[e^{\mathcal{H}_{\theta}}]} \dyad{j}\right) \int_{-\infty}^{\infty} f(t) \left( \sum_k e^{it\lambda_k} \dyad{k} \right) W \left(\sum_l e^{-it\lambda_l} \dyad{l}\right) dt W ]  \\
    &= \Re \sum_{j,k,l} \frac{e^{\lambda_j}}{\Tr[e^{\mathcal{H}_{\theta}}]} \int_{-\infty}^{\infty} f(t) e^{it\lambda_k} e^{-it\lambda_l} dt \braket{j}{k} W_{kl} W_{lj} \\
    &= \Re \sum_{j,l} \frac{e^{\lambda_j}}{\Tr[e^{\mathcal{H}_{\theta}}]} \int_{-\infty}^{\infty} f(t) e^{-it(\lambda_l - \lambda_j)} dt W_{jl} W_{lj} \\
    &= \Re \sum_{j,l} \frac{e^{\lambda_j}}{\Tr[e^{\mathcal{H}_{\theta}}]} \frac{\tanh(\frac{\lambda_l - \lambda_j}{2})}{\frac{\lambda_l - \lambda_j}{2}} |W_{jl}|^2 \\
    &= \sum_{j} \frac{e^{\lambda_j}}{\Tr[e^{\mathcal{H}_{\theta}}]} |W_{jj}|^2 + \sum_{j, l \neq j} \frac{e^{\lambda_j}}{\Tr[e^{\mathcal{H}_{\theta}}]} \frac{\tanh(\frac{\lambda_l - \lambda_j}{2})}{\frac{\lambda_l - \lambda_j}{2}} |W_{jl}|^2 \\
    &\geq \left( \sum_{j} \frac{e^{\lambda_j}}{\Tr[e^{\mathcal{H}_{\theta}}]} W_{jj} \right)^2 .
\end{split}
\end{align}
Here we use the definition of the quantum belief propagation operator to evaluate the integral. The last step follows from Jensen's inequality applied to the first term, and by ignoring the second term (which is always positive).
Plugging this in Eq.~\eqref{eq:positive_semidefinite_0} we obtain the desired result
\begin{align}
\label{eq:positive_semidefinite}
    v^T \nabla^2 S v &= \Re \Tr[\rho_\theta \Phi(W) W ] - \Tr[ \rho_\theta W]^2 
    \geq \left( \sum_{j} \frac{e^{\lambda_j}}{\Tr[e^{\mathcal{H}_{\theta}}]} W_{jj} \right)^2 - \Tr[ \rho_\theta W]^2 = 0 .
\end{align}
We conclude that the quantum relative entropy is convex. We now show strict convexity by a contradiction argument, following Proposition 17 in~\cite{Anshu_2021}. Assume we have found one set of parameters $\theta^*$ with $\nabla S(\eta \|\rho_{\theta^*})=0$. Then from Eq.~\eqref{eq:deriv_alt} we have $\langle H_{i}\rangle_\eta =\langle H_{i} \rangle_{\rho_\theta^*}$ for all $H_{i}$. Note that we can always find at least one such $\theta^*$ by Jaynes' principle~\cite{jaynes1957}. Next, assume there exists a different set of parameters, $\chi\neq\theta^*$, with $\langle H_{i}\rangle_\eta =\langle H_{i} \rangle_{\rho_\chi}$ for all $H_{i}$. Then
\begin{align}
\begin{split}
S(\rho_\chi \| \rho_{\theta^*}) &= \Tr[ \rho_\chi\log{\rho_\chi} ]-\Tr[ \rho_\chi\log{\rho_{\theta^*}} ] \\
&= \Tr[\rho_\chi\log{\rho_\chi}] - \sum_{i}\theta^*_{i}\Tr[\rho_\chi H_{i}] +\log{Z_{\theta^*}} \\
&= \Tr[\rho_\chi\log{\rho_\chi}] - \sum_{i}\theta^*_{i}\Tr[\rho_{\theta^*} H_{i}] +\log{Z_{\theta^*}} \\
&= \Tr[\rho_\chi\log{\rho_\chi}] - \Tr[\rho_{\theta^*}\log{\rho_{\theta^*}}] \\
&\geq 0. 
\end{split}
\end{align} 
Similarly, by swapping $\rho_\chi$ and $\rho_{\theta^*}$, we find \begin{equation}
S(\rho_{\theta^*}\|\rho_\chi) = \Tr[\rho_{\theta^*}\log{\rho_{\theta^*}}] - \Tr[\rho_\chi\log{\rho_\chi}] \geq 0.
\end{equation} 
It follows that $S(\rho_\theta^*\|\rho_\chi)=0$, implying $\rho_\theta^* = \rho_\chi$. Since the operators $H_{i}$ are orthogonal we have that $\theta^*=\chi$. But this contradicts the assumption made at the beginning ($\theta^*\neq\chi$). Thus we can only have one unique $\theta^*$ such that $\nabla S(\eta \|\rho_{\theta^*})=0$, i.e., $S$ is strictly convex by Definition~\ref{def:convexity}.
\end{proof}

\subsection*{Strong convexity}
\addcontentsline{toc}{subsection}{Strong convexity}

To show $\alpha$-strong convexity of $S$ one can use Lemma~\ref{lem:strong_convexity}. To the best of our knowledge, there is no proof in the literature showing that the quantum relative entropy of Gibbs states is strongly convex in general. On the other hand, this property has been proven for particular classes of Hamiltonians. Anshu et al.~\cite{Anshu_2021} proves strong convexity for $k$-local Hamiltonians defined on a finite dimensional lattice. They show that in this case $\alpha \in \Omega(1/n)$, a polynomial decrease with respect to the system size. Haah et al.~\cite{Haah_2024} proves strong convexity for the more general class of low-intersection Hamiltonians. Low-intersection Hamiltonians have terms that act non-trivially only on a constant number of qubits, and each term intersects non-trivially with a constant number of other terms.

In this section, we use differentiable programming~\cite{Baydin2017} to numerically analyze the smallest eigenvalue of the Hessian, $\lambda_{\min}(\nabla^2 S)$, seeking evidence for strong convexity. We consider a 1D nearest-neighbor Hamiltonian, $\mathcal{H} = \sum_{i=1}^{n-1} h_{i,i+1}^{xx} X_i X_{i+1} + h_{i, i+1}^{yy} Y_i Y_{i+1} + h_{i,i+1}^{zz} Z_i Z_{i+1}$, and a fully-connected one $\mathcal{H} = \sum_{i=1}^{n-1} \sum_{j > i}^n h_{i,j}^{xx} X_i X_j + h_{i,j}^{yy} Y_i Y_j + h_{i,j}^{zz} Z_i Z_j$. We randomly sample coefficients uniformly in $[-\mu, \mu]$ where $\mu$ is a scale parameter and determines the max-norm of the vector of the coefficients. Supplementary Figure~\ref{fig:min_eval_hess} shows mean and standard deviation over $25$ instances. The smallest eigenvalue decreases with the number of qubits until convergence to a fixed value. The fully-connected Hamiltonian has $m \in O(n^2)$ parameters and yields smaller eigenvalues than the 1D Hamiltonian which has $m \in O(n)$ parameters instead. These results provide evidence of strong convexity with $\alpha$ decreasing polynomially with the system size. We leave more extensive numerical studies for future work.

\begin{figure}[ht]
    \centering
    \includegraphics[width=.75\textwidth]{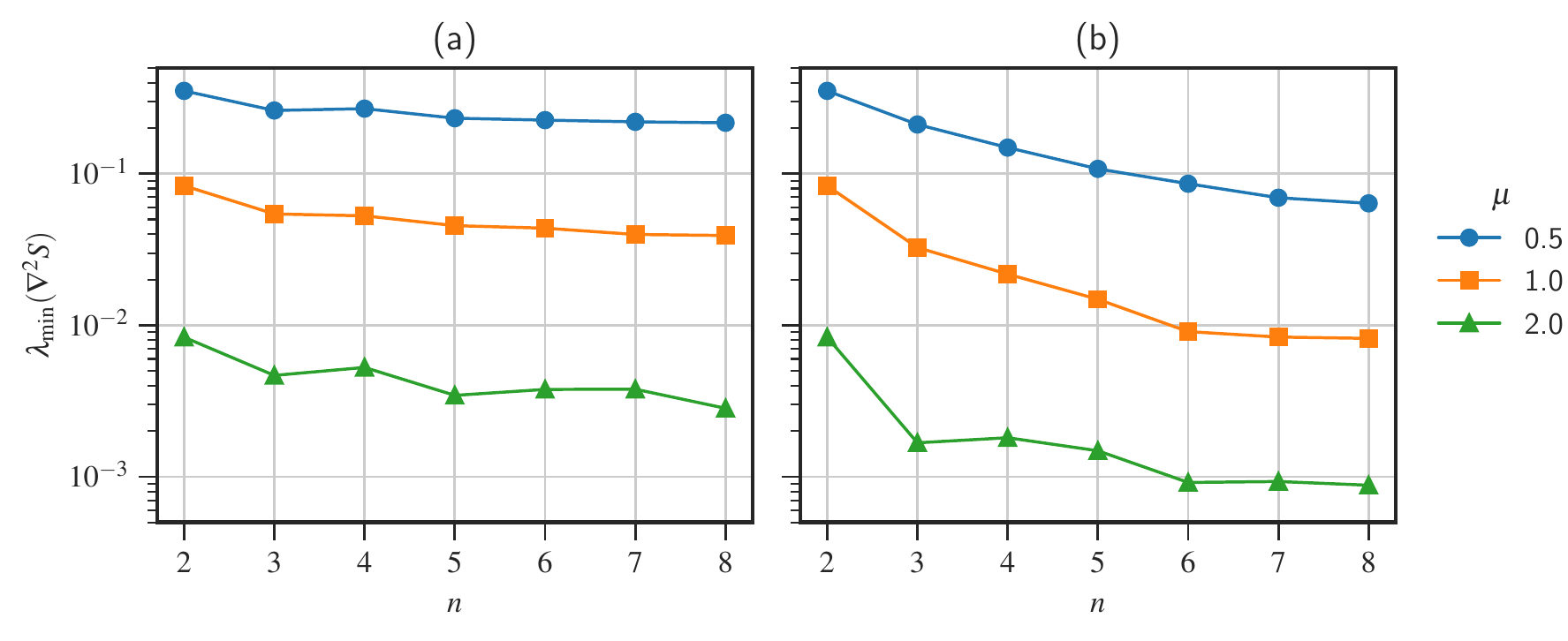}
    \caption{Minimum eigenvalue of the Hessian, as a function of the number of qubits. We show the median of $25$ random instances for the 1D nearest-neighbor Hamiltonian (a), and fully-connected Hamiltonian (b). The scale parameter $\mu$ determines the maximum size of random parameters. We observe that in all cases the smallest eigenvalues shrink with the number of qubits, but appear to plateau to a positive value.}
    \label{fig:min_eval_hess}
\end{figure}

\subsection*{\texorpdfstring{$L$}{L}-Smoothness}
\addcontentsline{toc}{subsection}{\texorpdfstring{$L$}{L}-Smoothness}

\begin{lemma}
Let $\mathcal{H}_\theta = \sum_{i=1}^m \theta_{i} H_{i}$ be a Hamiltonian ansatz where $H_i$ are non-commuting operators, and $\rho_\theta = \frac{e^{\mathcal{H}_\theta}}{Z}$ the corresponding QBM. For all states $\eta$, the quantum relative entropy $S(\eta\|\rho_\theta)$ is a L-smooth function of $\theta$ with $L=2m \max_j \| H_j \|_2^2$. 
\end{lemma}

\begin{proof}
We look for an upper bound on the largest eigenvalue of the Hessian in Eq.~\eqref{eq:hessian_new}. 
We begin with the following property
\begin{align}
\label{eq:phi_norm}
    \| \Phi(V) \| = \Big \| \sum_{j,k} \ketbra{j}{k} V_{jk} \hat{f}(\lambda_k - \lambda_j) \Big \| \leq | \hat{f}_{\max} | \Big \| \sum_{j,k} \ketbra{j}{k} V_{jk} \Big \| = \| V \| ,
\end{align}
where we use that $\hat{f}>0$ and $\hat{f}_{\max} = 1$. In what follows, we make use of this bound with the spectral norm $\| \cdot \|_2$, which is the operator norm induced by the Euclidean vector norm ($p=2$). 
Let us bound the entries of the Hessian
\begin{align}
\begin{split}
    \left| \frac{\partial^2 S}{\partial \theta_{j} \partial \theta_{k}} \right| &= \left | \frac{1}{2} \Tr[ \rho_\theta \{ \Phi(H_j), H_k \}] - \Tr[ \rho_\theta H_j] \Tr[ \rho_\theta H_k] \right| \\
    &\leq \frac{1}{2} \| \{\Phi(H_j), H_k\} \|_2 +  \| H_j \|_2 \| H_k \|_2 \\ 
    &\leq \| \Phi(H_j) \|_2 \| H_k \|_2 +  \| H_j \|_2 \| H_k \|_2 \\ 
    &\leq 2 \| H_j \|_2 \| H_k \|_2 .
\end{split}
\end{align}
Here we use the triangle inequality, that expectations are bounded by the spectral norm, the sub-multiplicative property of the spectral norm, and Eq.~\eqref{eq:phi_norm}. We are now able to upper-bound the largest eigenvalue $\lambda_\mathrm{max}$ of the Hessian
\begin{align}
\begin{split}
\label{eq:hess_upper_bound}
\| \nabla^2 S \|_2 &= | \lambda_\mathrm{max}(\nabla^2 S)| \\
&\leq \max_j \sum_{k=1}^m  \left| \frac{\partial^2 S}{\partial \theta_{j} \partial \theta_{k}} \right| \\
&\leq m \max_j \max_k  \left| \frac{\partial^2 S}{\partial \theta_{j} \partial \theta_{k}} \right| \\
&\leq 2m \max_j \max_k \| H_j \|_2 \| H_k \|_2 \\
&= 2 m \max_j \| H_j \|_2^2 .
\end{split}
\end{align}
The first inequality is a consequence of the Gershgorin circle theorem. We can now prove the $L$-smoothness of the quantum relative entropy. Let us define a function $h(t) = \nabla S(\eta \|\rho_{y + t(x - y)})$. Then we have
\begin{align}
\begin{split}
\| \nabla S(\eta \|\rho_{x}) - \nabla S(\eta \|\rho_{y}) \|_2 &= \| h(1) - h(0) \|_2 \\
&= \left\| \int_0^1 h^\prime(t) dt \right\|_2 \\
&\leq \int_0^1 \| h^\prime (t) \|_2 dt \\ 
&= \int_0^1 \| \nabla^2 S(\eta \|\rho_{y + t(x - y)}) \; (x - y) \|_2 dt  \\
&\leq \int_0^1 \| \nabla^2 S(\eta \|\rho_{y + t(x - y)}) \|_2 \; \| x - y \|_2  dt \\
&\leq 2m \max_j \| H_j \|_2^2 \; \| x - y \|_2 ,
\end{split}
\end{align}
where in the last step we used Eq.~\eqref{eq:hess_upper_bound}. Thus the quantum relative entropy is $L$-smooth with $L=2m \max_j \| H_j \|_2^2$.
\end{proof}

\section*{\MakeUppercase{Supplementary Note 3: Convergence results of SGD for training QBMs}}
\addcontentsline{toc}{section}{Supplementary Note 3: Convergence results of SGD for training QBMs}

In this appendix we first review useful results from the machine learning literature, then prove Theorems~\ref{appthm:training} and \ref{appthm:training_alpha} in the main text. We also discuss a few upper bounds for the relative entropy in the context of QBM learning. Finally, we provide additional numerical results for the scaling behavior of our theorems.

\subsection*{Review of SGD convergence results}
\addcontentsline{toc}{subsection}{Review of SGD convergence results}

We begin by stating three convergence results from the Stochastic Gradient Descent (SGD) literature. Consider a loss function $f : \mathbb{R}^m \rightarrow \mathbb{R}$ that is $L$-smooth (Def.~\ref{def:smooth}) and bounded from below by $f_\mathrm{inf} \in \mathbb{R}$. The stochastic gradient is unbiased, i.e., $\mathbb{E}[\hat{g}] = \nabla f$, and satisfies 
\begin{equation}
    \mathbb{E} \| \hat{g}(x) \|^2  \leq 2A (f(x) - f_\mathrm{inf}) + B \| \nabla f(x) \|^2 + C ,
\end{equation}
for some $A$, $B$, $C \geq 0$ and all $x \in \mathbb{R}^m$.
SGD iteratively minimizes $f$ according to the update rule $x^t = x^{t-1} - \gamma^t \hat{g}_{x^{t-1}}$ at time step $t$. Khaled and Richt{\'a}rik~\cite{Khaled2020} proved the following SGD convergence result.

\begin{lemma}[restatement of Corollary $1$ in~\cite{Khaled2020}] 
\label{lem:sgd_conv_1}
Choose precision $\epsilon > 0$ and step size $\gamma = \min
\{ \frac{1}{\sqrt{LAT}}, \frac{1}{LB}, \frac{\epsilon}{2LC} \} $, and set $\delta_0 = \mathbb{E}[f(x^0)] - f_\mathrm{inf}$. Then provided that 
\begin{equation}
    T \geq \frac{12 \delta_0 L}{\epsilon^2} \max \left\{ B, \frac{12 \delta_0 A}{\epsilon^2}, \frac{2C}{\epsilon^2} \right\} , 
\end{equation}
we have that SGD converges with 
\begin{equation}
\mathrm{min}_{1\leq t \leq T} \; \mathbb{E}\|\nabla f(x^t)\| \leq \epsilon .
\end{equation}
\end{lemma}

Here $\mathbb{E}[\cdot]$ denotes the expectation with respect to $x^t$, which is a random variable due to the stochasticity in the gradient. Let us now consider a loss function which, in addition to the previous conditions, is also $\alpha$-Polyak-{\L}ojasiewicz (Def.~\ref{def:PL}). We consider the following iterative learning rate scheme for $\gamma_t$. 
\begin{lemma}[restatement of Lemma $3$ in~\cite{Khaled2020}]
\label{lem:lrscheme}
Consider a sequence $(r_t)_t$ satisfying $r_{t+1}\leq (1-a \gamma_t)r_t + c\gamma_t^2$, where $\gamma_t \leq \frac{1}{b}$ for all $t\geq 0$ and $a,c\geq 0$ with $a\leq b$. Fix $T>0$ and let $k_0=\lceil \frac{T}{2}\rceil$. Then choosing the stepsize as 
\begin{equation}
    \gamma_t = \begin{cases}
       \frac{1}{b}, & \mathrm{if} \hspace{1mm} T \leq \frac{b}{a} \hspace{1mm} \mathrm{or} \hspace{1mm} t<k_0,\\
       \frac{2}{a(s+t-k_0)}, & \mathrm{if }\hspace{1mm} T \geq \frac{b}{a} \hspace{1mm} \mathrm{and} \hspace{1mm} t>k_0,\\
    \end{cases}
\end{equation} with $s=\frac{2b}{a}$ gives $r_T \leq \exp{-\frac{aK}{2b}}r_0 + \frac{9c}{a^2T}$
\end{lemma}
For this learning rate scheme, Khaled and Richt{\'a}rik~\cite{Khaled2020} proved the following SGD convergence result.

\begin{lemma}[restatement of Corollary $2$ in~\cite{Khaled2020}]
\label{lem:sgd_conv_2}
Choose precision $\epsilon > 0$ and step size $\gamma_t$ following Lemma~\ref{lem:lrscheme} with $\gamma_t\leq \min \{\frac{\alpha}{2AL}, \frac{1}{2BL}\}$ . Then provided that 
\begin{equation}
\label{eq:condition_1}
T \geq \frac{L}{\alpha} \max \left\{ \frac{2A}{\alpha}\log{\frac{2\delta_0}{\epsilon}}, \frac{2B}{\alpha}\log{\frac{2\delta_0}{\epsilon}}, \frac{9C}{2\alpha \epsilon} \right\}
\end{equation} 
we have that SGD converges with 
\begin{equation}
 \mathbb{E} |f(x^T)-f_{\mathrm{inf}} | \leq \epsilon .
\end{equation}
\end{lemma}

\subsection*{Proofs of Theorems~\ref{appthm:training} and~\ref{appthm:training_alpha} in the main text}
\addcontentsline{toc}{subsection}{Proofs of Theorems~\ref{appthm:training} and~\ref{appthm:training_alpha} in the main text}

We prove Theorem~\ref{appthm:training}, which is repeated here for completeness. 

\begingroup
\def\thetheorem{1}
\begin{theorem}[QBM training]\label{appthm:training}
Given a QBM defined by a set of $n$-qubit Pauli operators $\{H_i\}_{i=1}^m$, a precision $\kappa$ for the QBM expectations, a precision $\xi$ for the data expectations, and a target precision $\epsilon$ such that $\kappa^2 + \xi^2 \geq \frac{\epsilon}{2m}$. After 
\begin{equation}
    T = \frac{48\delta_0 m^2  (\kappa^2 + \xi^2 )}{\epsilon^4}
\end{equation} 
iterations of stochastic gradient descent on the relative entropy $S(\eta \|\rho_\theta)$ with constant learning rate $\gamma^t=\frac{\epsilon}{4m^2(\kappa^2+\xi^2)}$, we have
\begin{equation}
    \mathrm{min}_{t=1,..,T} \; \mathbb{E}| \expval{H_i}_{\rho_{\theta^t}}  - \expval{H_i}_\eta | \leq \epsilon, \qquad \forall i,
\end{equation} where $\mathbb{E}[\cdot]$ denotes the expectation with respect to the random variable $\theta^t$. 
Each iteration $t \in \{0, \dots, T\}$ requires
\begin{equation}
\label{eq:copies_app}
    N \in \mathcal{O} \left( \frac{1}{{\kappa^4}} \log \frac{ m} { 1-\lambda^{\frac{1}{T}} } \right)
\end{equation}
preparations of the Gibbs state $\rho_{\theta^t}$, and the success probability of the full algorithm is $\lambda$. Here, $\delta_0 = S(\eta\|\rho_{\theta^0})-S(\eta\|\rho_{\theta^\mathrm{opt}})$ is the relative entropy difference with the optimal model $\rho_{\theta^\mathrm{opt}}$.
\end{theorem}
\endgroup

\begin{proof}
Recall from Supplementary Note 2 that the quantum relative entropy is $L$-smooth with $L=2 m \max_i \| H_i \|_2^2$ and that for Pauli operators $\| H_i \|_2 = 1$. Then, we can minimize the relative entropy by SGD and apply the convergence result in Lemma~\ref{lem:sgd_conv_1}. 

For the SGD algorithm, we need an unbiased gradient estimator with bounded variance. We recall that the gradient of the relative entropy is given by $\partial_{\theta_i} S(\eta \|\rho_\theta)  = \langle H_i \rangle_{\rho_\theta} - \langle H_i\rangle_\eta$. The target expectation values $\langle H_i\rangle_\eta$ are estimated as $\hat{h}_{i,\eta}$ from the data set, as described in Supplementary Note 5. Note that $| \langle H_i\rangle_\eta - \hat{h}_{i,\eta} | \leq \xi$, where $\xi >0$ is limited by the size of the data set. One can improve on $\xi$ by collecting more data, as long as the number of samples is polynomial in $n$.

For estimating the QBM expectation values $\langle H_i \rangle_{\rho_\theta}$, we can use a number of techniques. Here we focus on classical shadow tomography. From Theorem 4 in~\cite{Huang_2021}, there exists a procedure that returns the expectation values of $m$ different Pauli operators $\{H_i\}$ to precision $\kappa$ with $\mathcal{O}(\frac{\log{m/\tilde{\lambda}}}{\kappa^4})$ preparations of $\rho_\theta$, and succeeds with probability $1-\tilde{\lambda}$. To make use of this result we now restrict the $H_i$ in the QBM Hamiltonian to be Pauli operators. As discussed in the main text, this is not a limitation, and it is possible to use non-Pauli operators in combination with other shadow tomography protocols. We can therefore obtain estimators $\hat{h}_{i, \rho_\theta}$ such that 
\begin{equation}
\label{eq:shadows}
    \max_i |\hat{h}_{i, \rho_\theta} - \langle H_i\rangle_{\rho_\theta}| \leq \kappa. 
\end{equation} 

We then use $\hat{g}_{\theta_i} = \hat{h}_{i, \rho_\theta} - \hat{h}_{i, \eta}$ as estimators for the partial derivatives of the quantum relative entropy. The variance of the norm of the gradient estimator is bounded as
\begin{align}
\begin{split}
\label{eq:grad_bounded_variance}
    \mathbb{E} \|\hat{g}_\theta - \nabla S(\eta \|\rho_\theta) \|^2 &=
    \mathbb{E} \sum_{i=1}^m (\hat{h}_{i, \rho_\theta} - \hat{h}_{i, \eta} - \langle H_i \rangle_{\rho_\theta} + \langle H_i\rangle_\eta )^2 \\
    &\leq \mathbb{E}\sum_{i=1}^m (\hat{h}_{i, \rho_\theta} - \langle H_i \rangle_{\rho_\theta} )^2 + (\hat{h}_{i, \eta} - \langle H_i\rangle_\eta )^2 \\
    &\leq m ( \kappa^2 + \xi^2).
\end{split}
\end{align} 
Since the variance can also be written as $\mathbb{E} \|\hat{g}_\theta \|^2 - \| \nabla S(\eta \|\rho_\theta) \|^2$ we find that our setup is compatible with Eq.~\eqref{eq:condition_1} for $A=0, B=1, C=m ( \kappa^2 + \xi^2)$. We choose $\epsilon<1$ and $\kappa^2 + \xi^2 \geq \frac{\epsilon}{2m}$ in Lemma~\ref{lem:sgd_conv_1}. This yields a learning rate of $\gamma = \frac{\epsilon}{ 4 m^2 (\kappa^2 + \xi^2)}$. We conclude that after 
\begin{equation}
    T \geq \frac{48 \delta_0 m^2 ( \kappa^2 + \xi^2)}{\epsilon^4}
\end{equation}
iterations of SGD we have
\begin{equation}
    \mathrm{min}_{1\leq t \leq T} \; \mathbb{E}\|\nabla S(\eta \|\rho_{\theta^t}) \| \leq \epsilon. 
\end{equation} 
Here $\delta_0 = S(\eta\|\rho_{\theta^0})-S(\eta\|\rho_{\theta^\mathrm{opt}})$ is the relative entropy at the initialization minus the relative entropy at the optimum. Importantly we note the QBM expectation values are computed with a success probability $1-\tilde{\lambda}$ at each iteration. Consequently, the total success probability of the whole training is equal to $(1-\tilde{\lambda})^T$ for $T$ update steps. Then to have a total success probability of $\lambda$ we need to set $\tilde{\lambda}=(1-\lambda^{1/T})$ in the shadow tomography protocol. This result, together with the sampling bound on the number of measurements of the shadow tomography, $\mathcal{O}(\frac{\log{m/\tilde{\lambda}}}{\kappa^4})$, completes the proof of Theorem~\ref{appthm:training}.
\end{proof}

We now provide a proof for Theorem~\ref{appthm:training_alpha}, which we restate here. 

\begingroup
\def\thetheorem{2}
\begin{theorem}[$\alpha$-strongly convex QBM training]\label{appthm:training_alpha}
Given a QBM defined by a Hamiltonian ansatz $\mathcal{H}_\theta$ such that $S(\eta\|\rho_\theta)$ is $\alpha$-strongly convex, a precision $\kappa$ for the QBM expectations, a precision $\xi$ for the data expectations, and a target precision $\epsilon$ such that $\kappa^2 + \xi^2 \geq \frac{\epsilon}{2m}$. After
\begin{equation}
    T =\frac{18 m^2 (\kappa^2 + \xi^2)}{\alpha^2\epsilon^2}
\end{equation} 
iterations of stochastic gradient descent on the relative entropy $S(\eta \|\rho_\theta)$ with learning rate $\gamma^t \leq \frac{1}{4m^2}$, we have
\begin{equation}
    \mathrm{min}_{t=1,..,T} \; \mathbb{E}| \expval{H_i}_{\rho_{\theta^t}}  - \expval{H_i}_\eta | \leq \epsilon, \qquad \forall i.
\end{equation} 
Each iteration requires the number of samples given in Eq.~\eqref{eq:copies_app}. 
\end{theorem} 
\endgroup

\begin{proof}
In order to prove this theorem, we first show that $\eta$, $\rho^\mathrm{opt}$, and $\rho_\theta$ are `collinear' with respect to the relative entropy. This is helpful because it allows us to directly bound the difference in relative entropy $ S(\eta \| \rho_{\theta}) - S(\eta \| \rho_{\theta^\mathrm{opt}})$ instead of bounding the individual relative entropies.
\begin{align}
\begin{split}
\label{eq:collinear}
    S(\eta \| \rho_{\theta}) - S(\eta \| \rho_{\theta^\mathrm{opt}})&= -\Tr[\eta \log \rho_\theta ] + \Tr[ \eta \log \rho_{\theta^\mathrm{opt}}] \\
    &= -\Tr[\eta \mathcal{H}_\theta] + \log Z_\theta + \Tr[\eta \mathcal{H}_{\theta^\mathrm{opt}}] - \log Z_{\theta^\mathrm{opt}} \\
    &= -\Tr[\rho_{\theta^\mathrm{opt}} \mathcal{H}_\theta] + \log Z_\theta + \Tr[\rho_{\theta^\mathrm{opt}} \mathcal{H}_{\theta^\mathrm{opt}}] - \log Z_{\theta^\mathrm{opt}} \\
    &= -\Tr[\rho_{\theta^\mathrm{opt}} \log \rho_{\theta} ] + \Tr[\rho_{\theta^\mathrm{opt}} \log \rho_{\theta^\mathrm{opt}} ]  \\
    &= S(\rho_{\theta^\mathrm{opt}} \| \rho_{\theta}) .
\end{split}
\end{align}
Here, in the going from the second to the third line, we used the fact that $\Tr[\eta H_i] = \Tr[\rho_{\theta^{\mathrm{opt}}} H_i]$, which follows from setting Eq.~\eqref{eq:deriv_alt} to zero. Rearranging the terms we get the collinearity $S(\eta \| \rho_{\theta}) = S(\eta \| \rho_{\theta^\mathrm{opt}}) + S(\rho_{\theta^\mathrm{opt}} \| \rho_{\theta})$.  
This is a non-trivial result because the relative entropy is not a distance: it is not symmetric and does not satisfy the triangle inequality in general. With this relation, we are now able to prove Theorem~\ref{appthm:training_alpha}. 

The relative entropy $S(\rho_{\theta^\mathrm{opt}} \| \rho_\theta)$ satisfies all the relevant assumptions for SGD convergence: it is an $L$-smooth function with $L=2 m \max_i \| H_i \|_2^2$, it is bounded below by $0$, and the stochastic gradient has bounded variance [Eqs.~\eqref{eq:grad_bounded_variance} and~\eqref{eq:copies_app} apply]. In addition, the $\alpha$-strong convexity assumed by the theorem implies that $S(\rho_{\theta^\mathrm{opt}} \| \rho_\theta)$ is $\alpha$-Polyak-{\L}ojasiewicz by Lemma~\ref{lem:PL}. This means we can invoke Lemma~\ref{lem:sgd_conv_2}. As before, we set $A=0, B=1, C=m( \kappa^2 + \xi^2)$ and choose $\epsilon^\prime < 1$ in the Lemma, thus obtaining a maximum learning rate $\gamma_t \leq \frac{1}{4m^2}$. Looking at the case where $\frac{2}{\alpha}\log{\frac{2\delta_0}{\epsilon^\prime}}\leq\frac{9m(\kappa^2 + \xi^2)}{2\alpha \epsilon^\prime}$ we find that after $T \geq \frac{9m^2(\kappa^2 + \xi^2)}{\alpha^2\epsilon^\prime}$ iterations the expected relative entropy is $\mathbb{E} S(\rho_{\theta^\mathrm{opt}} \|\rho_{\theta^T}) \leq \epsilon^\prime$. Note that, depending on the problem-specific parameter $\delta_0$, and the free parameters $\kappa$ and $\xi$, one could be in the other case of Lemma~\ref{lem:sgd_conv_2}. Following the same steps shown here one arrives at a different number of iterations $T$, which is still polynomial in $n$.\\

Since the QBM learning problem is phrased in terms of the distance in the expectation values and we only have obtained a bound on the relative entropy we now relate the two. Using Pinsker's and Jensen's inequalities it follows that
\begin{align}
\begin{split}
    \mathbb{E} S(\rho_{\theta^\mathrm{opt}} \|\rho_{\theta^T}) &\geq \frac{1}{2 \ln 2} \mathbb{E} \| \rho_{\theta^\mathrm{opt}} - \rho_{\theta^T}  \|^2_1 \\
    &\geq \frac{1}{2} ( \mathbb{E} \| \rho_{\theta^\mathrm{opt}} - \rho_{\theta^T}  \|_1 )^2 \\
    &= \frac{1}{2} ( \mathbb{E} \max_{-I < U \leq I} | \Tr[ U (\rho_{\theta^\mathrm{opt}} - \rho_{\theta^T}) ] | )^2,
\end{split}
\end{align}
where in the last line we used the variational definition of trace distance. The maximization is over unitary matrices. Let us restrict the maximization to unitary matrices defined as $U_i = \frac{H_i}{\| H_i \|} + i \sqrt{ I - \frac{H_i^2}{\| H_i \|^2} }$. These have the property that $H_i = \| H_i \| (U_i + U_i^\dag )/2$. 
Therefore,
\begin{align}
\begin{split}
    \sqrt{2\mathbb{E} S(\rho_{\theta^\mathrm{opt}} \|\rho_{\theta^T})} &\geq \mathbb{E} \max_{-I < U \leq I} \left| \Tr[ \frac{1}{2}(U + U^\dag) (\rho_{\theta^\mathrm{opt}} - \rho_{\theta^T}) ] \right| \\
    &\geq \mathbb{E} \max_i \left| \Tr[ \frac{1}{2}(U_i + U_i^\dag)(\rho_{\theta^\mathrm{opt}} - \rho_{\theta^T}) ] \right| \\
    & \geq \mathbb{E} \max_i \frac{1}{\| H_i\|} | \Tr[ H_i \rho_{\theta^\mathrm{opt}}]  - \Tr[ H_i \rho_{\theta^T} ] | .
\end{split}
\end{align}
This implies,
\begin{align}
    \sqrt{2\epsilon^\prime} \geq \frac{1}{\| H_i \| } \mathbb{E} | \Tr[ H_i \rho_{\theta^\mathrm{opt}}]  - \Tr[ H_i \rho_{\theta^T} ] | , \quad \forall i . 
\end{align}
To solve the QBM learning problem to precision $\epsilon$ we choose $\epsilon^\prime = \frac{\epsilon^2}{2}$ and, since $\| H_i \| = 1$ for Pauli operators, we conclude that
\begin{align}
    \mathbb{E} | \Tr[ H_i \eta]  - \Tr[ H_i \rho_{\theta^T} ] | \leq \epsilon , \quad \forall i . 
\end{align}
\end{proof}

\subsection*{Achieving a desired precision on the quantum relative entropy}
\addcontentsline{toc}{subsection}{Achieving a desired precision on the quantum relative entropy}

In this section, we study the scenario where the reader is interested in obtaining a certain precision on the quantum relative entropy, rather than on the difference in the expectation values. Again, due to a potential model mismatch, we discuss the relative entropy $S(\rho_{\theta^\mathrm{opt}} \| \rho_{\theta})$ w.r.t. the optimal QBM $\rho_{\theta^\mathrm{opt}}$.

We begin by training the QBM $\rho_{\theta}$ with SGD. Using Theorem~\ref{appthm:training}, we can achieve $|\langle H_{i}\rangle_\eta -\langle H_{i} \rangle_{\rho_\theta} | \leq \epsilon$ for all $i$ with polynomial sample complexity. This implies a similar relation w.r.t. the optimal model: $| \langle H_{i}\rangle_\eta -\langle H_{i} \rangle_{\rho_{\theta^\mathrm{opt}}} | \leq \epsilon$. By the triangle inequality we have that $| \langle H_{i}\rangle_{\rho_\theta} - \langle H_{i} \rangle_{\rho_{\theta^\mathrm{opt}}} | \leq 2\epsilon$.
Then
\begin{align}
\begin{split}
S(\rho_{\theta^\mathrm{opt}} \| \rho_\theta) &=\Tr[\rho_{\theta^\mathrm{opt}} \log \rho_{\theta^\mathrm{opt}}] - \sum_i \theta_i \Tr[\rho_{\theta^\mathrm{opt}} H_i] + \log{Z_{\theta}} \\
&= \Tr[\rho_{\theta^\mathrm{opt}} \log \rho_{\theta^\mathrm{opt}}] - \sum_i \theta_i \Tr[\rho_{\theta^\mathrm{opt}} H_i] + \log Z_\theta + \sum_i \theta_i \Tr[\rho_{\theta} H_i] -  \sum_i \theta_i \Tr[\rho_{\theta} H_i] \\
&= \Tr[\rho_{\theta^\mathrm{opt}} \log \rho_{\theta^\mathrm{opt}}] - \Tr[\rho_{\theta} \log \rho_{\theta}] + \sum_i \theta_i ( \Tr[\rho_{\theta} H_i] - \Tr[\rho_{\theta^\mathrm{opt}} H_i] ).
\end{split}
\end{align}
Similarly,
\begin{align}
S(\rho_{\theta} \| \rho_{\theta^\mathrm{opt}}) =  - \Tr[\rho_{\theta^\mathrm{opt}} \log \rho_{\theta^\mathrm{opt}}] + \Tr[\rho_{\theta} \log \rho_{\theta}] - \sum_i \theta^\textrm{opt}_i ( \Tr[\rho_{\theta} H_i] - \Tr[\rho_{\theta^\mathrm{opt}} H_i] ). 
\end{align}
Thus
\begin{align}
\begin{split}
\label{eq:up_bound_for_S}
S(\rho_{\theta^\mathrm{opt}} \| \rho_\theta) &\leq S(\rho_{\theta^\mathrm{opt}} \| \rho_\theta) + S(\rho_{\theta} \| \rho_{\theta^\mathrm{opt}}) \\
&= \sum_i ( \Tr[\rho_{\theta} H_i] - \Tr[\rho_{\theta^\mathrm{opt}} H_i] ) ( \theta_i - \theta^\textrm{opt}_i) \\
&\leq \sum_i | \Tr[\rho_{\theta} H_i] - \Tr[\rho_{\theta^\mathrm{opt}} H_i] | \cdot | \theta_i - \theta^\textrm{opt}_i | \\
&\leq 2 \epsilon \| \theta - \theta^\textrm{opt} \|_1 .
\end{split}
\end{align}

To minimize the quantum relative entropy to precision $\epsilon^\prime$, we choose $\epsilon \leq \frac{\epsilon^\prime}{ 2 \| \theta - \theta^\textrm{opt} \|_1 }$. This determines the number of SGD iterations via Theorem~\ref{appthm:training}. Note that the number of iterations remains polynomial in the system size $n$. Finally, we combine this result with Eq.~\eqref{eq:collinear} and obtain the implication
\begin{align}
\label{eq:app_relentropbound}
    |\langle H_{i}\rangle_\eta -\langle H_{i} \rangle_{\rho_\theta} | \leq \epsilon, \forall i \quad \implies \quad S(\eta \| \rho_{\theta}) - S(\eta \| \rho_{\theta^\mathrm{opt}}) = S(\rho_{\theta^\mathrm{opt}} \| \rho_{\theta}) \leq 2 \epsilon \| \theta - \theta^\textrm{opt} \|_1 .
\end{align} 
This proves Eq.~\eqref{eq:relentropbound} in the main text. 

\subsection*{Additional numerical results}
\addcontentsline{toc}{subsection}{Additional numerical results}

Here we provide additional numerical simulation results which verify the scaling behaviors proven in Theorems~\ref{appthm:training} and~\ref{appthm:training_alpha}. In order to obtain robust estimates one would need to perform many repetitions of SGD training and analyze different settings of the QBM model and target state. However, to keep the computational cost of the simulations manageable, we only consider a few configurations of system size, learning rate, and ansatz. We focus on the target state constructed from the expectation values of the XXZ Heisenberg model (see Supplementary Note 5).
Unless stated otherwise, we use a fully-connected QBM model with Hamiltonian ansatz $\mathcal{H}_\theta = \sum_{k=x,y,z}\sum_{i, j>i}\lambda^k_{ij}\sigma_i^k\sigma^k_j + \sum_i^n \gamma^k_i\sigma^k_i$, where the number of terms is quadratic in the system size, $m \in \mathcal{O}(n^2)$. For each configuration, we run SGD training only once, starting from the maximally mixed state.

Supplementary Figure~\ref{fig:scale}~(a) shows the required number of SGD iterations, $T$, to reach convergence for different system sizes. Convergence is defined as the first iteration where the QBM matches the target expectation values up to precision $\epsilon=0.1$, i.e., when the QBM learning problem is solved to precision $\epsilon$. We use a combined model and data expectations precision of $\kappa^2+\xi^2=0.01$ which is larger than $\frac{\epsilon}{2m}$ as required by our theorems. We compare training with the learning rate from Theorem~\ref{appthm:training} to training with a constant learning rate $\gamma_t=0.1\epsilon$. We observe that for the learning rate from Theorem~\ref{appthm:training} (blue line) the number of iterations $T$ grows with the system size. From the log-log plot in the inset we find evidence that the scaling follows a power law with a power (slope) of $\approx3.3$. This indicates that $T\propto n^{3.3}$, in agreement with the upper bound from Theorem~\ref{appthm:training} applied to our model, $T \propto m^2 \in \mathcal{O}(n^4)$. When using a constant learning rate we observe a much more rapid convergence (orange line). In particular, $T$ does not seem to increase with the system size. However, there are no proven convergence guarantees for this setting. That is, for very noisy gradients and/or high target precision, a constant learning rate may lead to unstable training. In practice one could opt for an in-between learning rate schedule or use an optimizer with adaptive learning rates, such as the well-known ADAM optimizer. 

Supplementary Figure~\ref{fig:scale}~(b) shows the required number of iterations versus the target precision, $\epsilon^{-1}$. For these simulations we use $\kappa^2+\xi^2=0.1$ and system size $n=6$. Using a fully-connected QBM, we compare training with the learning rate from Theorem~\ref{appthm:training} (blue line) to training with the learning rate from Theorem~\ref{appthm:training_alpha} (orange line). Recall that we provided evidence for the strong convexity of the quantum relative entropy of fully-connected QBMs in Supplementary Note~2. 
If the loss landscape is indeed strongly convex, we can use the learning rate from Theorem~\ref{appthm:training_alpha} and the number of steps scales as $1/\epsilon^2$ otherwise it is expected to scale as $1/\epsilon^4$.
Thus, we expect the learning rate from Theorem~\ref{appthm:training_alpha} to outperform the learning rate from Theorem~\ref{appthm:training}. Indeed we observe convergence in a low degree polynomial number of iterations for the orange line, versus a higher degree polynomial number of iterations for the blue line. Since the strong convexity of fully-connected QBMs is unproven, it is hard to justify the choice of learning rate. We additionally tested the scaling using a 1D nearest-neighbor Hamiltonian ansatz for the QBM (green line), which is known to be strongly convex as discussed in Supplementary Note 2.
The green line shows a very low degree polynomial scaling when using the learning rate schedule from Theorem~\ref{appthm:training_alpha}. These numerical results match the expected scaling from our theorems. 

Lastly, here and in the main text we focused on the error in the expectation values. If one is interested in other performance metrics, they can employ Eq.~\eqref{eq:app_relentropbound} to convert our expectation value plots to relative entropy bounds.

\begin{figure}[ht]
    \centering
    \includegraphics[width=0.92\textwidth]{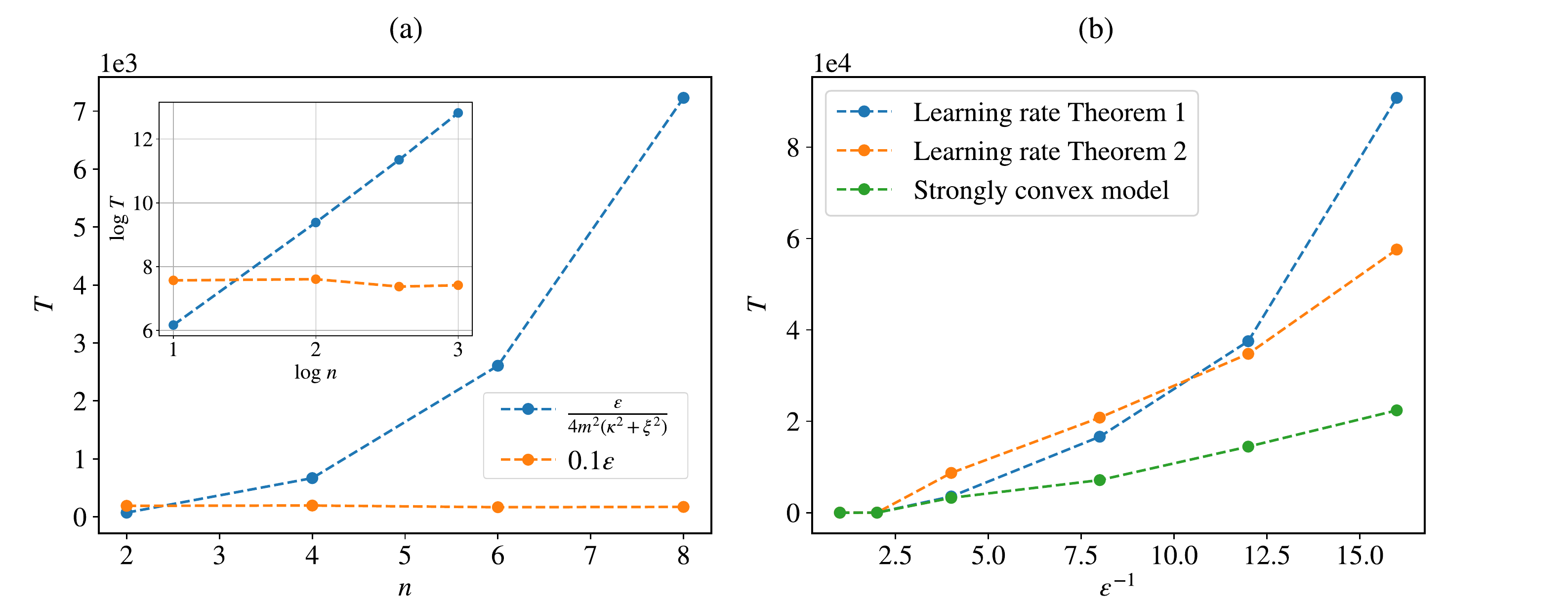}
    \caption{(a) Number of SGD iterations required to solve the QBM learning problem for precision $\epsilon=0.1$ versus system size. The learning rate from Theorem~\ref{appthm:training} is compared to a constant learning rate. The inset shows the log-log plot of the data in the main panel. (b) Number of SGD iterations to solve the QBM learning problem for system size $n=6$ versus the target precision $\epsilon^{-1}$. We compare a fully-connected QBM (blue and orange lines) to a 1D nearest-neighbor QBM (green line).}
    \label{fig:scale}
\end{figure}

\section*{\MakeUppercase{Supplementary Note 4: Guaranteed performance improvement by pre-training}}
\addcontentsline{toc}{section}{Supplementary Note 4: Guaranteed performance improvement by pre-training}

In this appendix, we first prove Theorem~\ref{appthm:pretraining} in the main text, and then discuss various pre-training models. 

\subsection*{Proof of Theorem~\ref{appthm:pretraining} in the main text}
\addcontentsline{toc}{subsection}{Proof of Theorem~\ref{appthm:pretraining} in the main text}

For completeness, we start by restating Theorem~\ref{appthm:pretraining} from the main text. 

\begingroup
\def\thetheorem{3}
\begin{theorem}[QBM pre-training]\label{appthm:pretraining}
    Assume a target $\eta$ and a QBM model $\rho_\theta=e^{\sum_{i=1}^m \theta_i H_i}/Z$ for which we like to minimize the relative entropy $S(\eta \|\rho_\theta)$. Initializing $\theta^0=0$ and pre-training $S(\eta\| \rho_\theta)$ with respect to any subset of $\tilde{m} \leq m $ parameters guarantees that 
    \begin{equation}
    \label{eq:pretrainrel_app}
         S(\eta \|\rho_{\theta^\mathrm{pre}}) \leq S(\eta \|\rho_{\theta^0}), 
    \end{equation} 
    where $\theta^\mathrm{pre} = [\chi^{\mathrm{pre}}, 0_{m-\tilde{m}}]$ and the vector $\chi^{\mathrm{pre}}$ of length $\tilde{m}$ contains the parameters for the terms $\{H_i\}_{i=1}^{\tilde{m}}$ at the end of pre-training.
    More precisely, starting from $\rho_\chi=e^{\sum_{i=1}^{\tilde{m}} \chi_i H_i}/Z$ and minimizing $S(\eta \|\rho_\chi)$ with respect to $\chi$ ensures Eq.~\eqref{eq:pretrainrel_app} for any $S(\eta \|\rho_{\chi^{\mathrm{pre}}})\leq S(\eta \|\rho_{\chi^0})$.
\end{theorem}
\endgroup

\begin{proof}
First, we relate the difference in relative entropy between two parameter vectors in the full space to the difference in relative entropy of the pre-trained parameter space. In particular, for any real parameter vectors $\theta = [\chi, 0_{m-\tilde{m}}]$ and $\theta' = [\chi', 0_{m-\tilde{m}}]$ we have
\begin{align}
\begin{split}
    S(\eta \|\rho_{\theta}) - S(\eta \|\rho_{\theta'}) &= \Tr\left[\eta \log \rho_{\theta'}\right] - \Tr\left[\eta \log \rho_\theta\right]\\
    &=\sum_{i=1}^m(\theta_i'-\theta_i)\Tr[\eta H_i] - \log\Tr[e^{\sum_{i=1}^m\theta_i'H_i}] + \log\Tr[e^{\sum_{i=1}^m\theta_iH_i}] \\
    &= \sum_{i=1}^{\tilde{m}}(\chi_i'-\chi_i)\Tr[\eta H_i] - \log\Tr[e^{\sum_{i=1}^{\tilde{m}}\chi_i'H_i}] + \log\Tr[e^{\sum_{i=1}^{\tilde{m}}\chi_iH_i}] \\
    &= \Tr[\eta \log \rho_{\chi'}] - \Tr[\eta \log \rho_\chi] \\
    &= S(\eta \|\rho_{\chi}) - S(\eta \|\rho_{\chi'})
\end{split}
\end{align}
Now using pre-training vectors $\theta^\mathrm{pre} = [\chi^{\mathrm{pre}}, 0_{m-\tilde{m}}]$ and $\theta^0 = [\chi^0, 0_{m-\tilde{m}}]=0$ we see that $S(\eta \|\rho_{\chi^{\mathrm{pre}}})\leq S(\eta \|\rho_{\chi^0})$ implies  $S(\eta \|\rho_{\theta^\mathrm{pre}}) \leq S(\eta \|\rho_{\theta^0})$. Thus, any method that finds such a $\chi^{\mathrm{pre}}$ guarantees Eq.~\eqref{eq:pretrainrel_app}.
\end{proof}

While conclusive, the above proof does not provide us with a method to find such a $\chi^{\mathrm{pre}}$, i.e., it is agnostic to the specific pre-training method. As a constructive example, let us consider minimizing $\chi^{\mathrm{pre}}$ with noiseless gradient descent on a subset of $\tilde{m}$ parameters. This means we update the subset parameters as $\chi^{t}=\chi^{t-1}-\gamma \tilde{\nabla}S(\eta \|\rho_{\chi^{t-1}}) $, where $\tilde{\nabla}S(\eta \|\rho_{\chi^{t-1}})$ is the gradient of the subset of parameters, and $\gamma$ the learning rate. Since $S$ is $L$-smooth, we can use the descent Lemma~\ref{lem:descent} to bound the difference in relative entropy of the subset
\begin{align}
\begin{split}
\label{eq:descentlemma}
S(\eta \|\rho_{\chi^{t}}) - S(\eta \|\rho_{\chi^{t-1}}) &\leq \nabla S(\eta \|\rho_{\chi^{t-1}})^T \left(-\tilde{\gamma}\nabla S(\eta \|\rho_{\chi^{t-1}}) \right) + \frac{L}{2} \| -\tilde{\gamma}\nabla S(\eta \|\rho_{\chi^{t-1}}) \|^2 \\
&= -\gamma\left(1-\frac{\gamma L}{2} \right) \|\tilde{\nabla} S(\eta \|\rho_{\chi^{t-1}}) \|^2.
\end{split}
\end{align} 
Setting $\gamma \leq \frac{2}{L}$ we obtain $S(\eta \|\rho_{\chi^{t}}) \leq S(\eta \|\rho_{\chi^{t-1}})$. By recursively applying this inequality we obtain a $\chi^{\mathrm{pre}}$ with $S(\eta \|\rho_{\chi^{\mathrm{pre}}}) \leq S(\eta \|\rho_{\chi^{0}})$, which by our theorem above ensures Eq.~\eqref{eq:pretrainrel_app}. Note that the smoothness $L$ here is the smoothness on the subset of parameters, which can be bounded by $L\leq 2\tilde{m} \max_i\| H_i\|_2^2$.

\subsection*{Pre-training models}
\addcontentsline{toc}{subsection}{Pre-training models}

Here we discuss possible pre-training models and strategies to optimize them. We focus on the models discussed in the main text: 1) a mean-field model, 2) a Gaussian Fermionic model, and 3) nearest-neighbor quantum spin models. The advantage of the first two models is that they can be trained analytically. While for the nearest-neighbor models, this is not possible, they satisfy the locality assumptions in \cite{Anshu_2021, Haah_2024}, and hence have a strongly convex relative entropy.

\subsubsection*{Mean-Field QBM}
\addcontentsline{toc}{subsubsection}{Mean-Field QBM}

We define the mean-field QBM by the parameterized Hamiltonian 
 \begin{equation}
    \mathcal{H}_\theta = \sum_i^n \theta^x_i \sigma_i^x + \theta^y_i \sigma_i^y +  \theta^z_i \sigma_i^z.
\end{equation} Since this Hamiltonian has a simple structure, in which many terms commute, we can find the optimal parameters analytically. First, recall that the QBM expectation values are given by \begin{equation}
    \langle H_i\rangle_{\rho_\theta} = \frac{\partial}{\partial \theta_i} \log{\Tr[e^{\mathcal{H}_\theta}]} =  \frac{\partial}{\partial \theta_i}\log{Z_\theta}. 
\end{equation} 
For the mean-field Hamiltonian, we find 
\begin{align}
    Z_\theta &= \Tr[e^{\sum_{i=1}^n\theta^x_i \sigma_i^x + \theta^y_i \sigma_i^y +  \theta^z_i \sigma_i^z}]
    =\prod_{i=1}^n \Tr[ e^{\theta^x_i \sigma_i^x + \theta^y_i \sigma_i^y + \theta^z_i \sigma_i^z}]
    = \prod_{i=1}^n 2\cosh{\|\theta_i\|_2},
\end{align} 
where we have defined $\|\theta_i\|_2=\sqrt{ {\theta_i^x}^2 + {\theta_i^y}^2 +{\theta_i^z}^2}$. Here we use the commutativity of single qubit operators in the first equality and expand the exponent for the second equality.  We therefore get \begin{equation}
\log{Z_\theta} = \sum^n_{i=1} \log{2\cosh{ \| \theta_i \|_2}}. 
\end{equation} 
From which the derivative follows as 
\begin{equation}
    \frac{\partial}{\partial \theta_i^{x,y,z}} \log{Z_\theta}
    = \frac{\theta_i^{x,y,z}}{\|\theta_i\|_2} \tanh{\|\theta_i\|_2}. 
\end{equation} 
In order to find the optimal QBM parameters for each qubit, $i$, we then solve the three coupled equations,
\begin{equation}
\label{eq:MFcoupled}
    \frac{\theta_i^{x,y,z}}{\|\theta_i\|_2} \tanh{\|\theta_i\|_2} = \langle \sigma_i^{x,y,z}\rangle_\eta,
\end{equation} 
which corresponds to setting the QBM derivative in Eq.~\eqref{eq:RelEntropDeriv} in the main text to zero. From the strict convexity of the relative entropy, we know this has one unique solution provided the target expectation values, $\langle \sigma_i^{x,y,z}\rangle_\eta$ form a consistent set, i.e., it comes from a density matrix. 
We can find the solution by squaring the three equations, and adding them together, giving
\begin{equation}
    \| \theta_i \|_2 = \tanh^{-1} \left(  \sqrt{\langle \sigma_i^{x}\rangle_\eta^2 + \langle \sigma_i^{y}\rangle_\eta^2 + \langle \sigma_i^{z}\rangle_\eta^2  } \right).
\end{equation}
Here we used that the argument of the tanh is always positive. Substituting this into Eq.~\eqref{eq:MFcoupled} we then find the closed-form solution of the QBM parameters
\begin{equation}
\theta_i^{x,y,z} = \langle \sigma_i^{x,y,z}\rangle_\eta \frac{ \tanh^{-1} \left( \sqrt{\langle \sigma_i^{x}\rangle_\eta^2 + \langle \sigma_i^{y}\rangle_\eta^2 + \langle \sigma_i^{z}\rangle_\eta^2  } \right) }{ \sqrt{\langle \sigma_i^{x}\rangle_\eta^2 + \langle \sigma_i^{y}\rangle_\eta^2 + \langle \sigma_i^{z}\rangle_\eta^2  } }.
\end{equation} 
In practice, the optimal parameters for an arbitrary mean-field QBM can be obtained by numerically evaluating this expression for the given target expectation values. 

\subsubsection*{Gaussian Fermionic QBM}
\addcontentsline{toc}{subsubsection}{Gaussian Fermionic QBM}

The Gaussian Fermionic QBM has a parameterized, quadratic, Fermionic Hamiltonian
\begin{equation}
    \mathcal{H}_\theta = \vec{C}^\dagger \tilde{\Theta} \vec{C} \equiv \sum_{i,j}\tilde{\Theta}_{ij} \vec{C}^\dagger_i \vec{C}_j.
\end{equation}
Here, $\vec{C}^\dagger = [c_1^\dagger, \dots, c_n^\dagger, c_1, \dots c_n ]$ is a vector containing $n$ Fermionic mode creation and annihilation operators, which satisfy the Fermionic commutation relations $\{c_i, c_j^\dagger\}=\delta_{i,j}$ and $\{c_i, c_j\}=0$. These Fermionic operators can be expressed as strings of Pauli operators by the Jordan-Wigner transformation. $\tilde{\Theta}$ is the  $2n\times 2n$ dimensional matrix containing the QBM model parameters $\theta$, which can be identified as a Fermionic single-particle Hamiltonian. Note that this matrix needs to be Hermitian, and since terms like $c^\dagger_1 c^\dagger_1$ are zero, it has in total $n^2$ free parameters. 

In order to find the optimal parameters, we use the fact that the single-particle correlation matrix with entries $ [\Gamma_{\rho_\theta}]_{ij} = \langle \vec{C}^\dagger_i \vec{C}_j \rangle_{\rho_\theta}$ contains sufficient information to compute all possible properties of the Gaussian quantum system. This includes all possible observables (via Wick's theorem), entanglement measures, and also sampling from $\rho_\theta$~\cite{Surace_2022}. In particular, the Gaussian Fermionic QBM gradient reduces to the difference in the correlation matrices of the target and the model
\begin{equation}
\frac{\partial S}{\partial \tilde{\Theta}_{ij}} = \langle \vec{C}^\dagger_i \vec{C}_j \rangle_{\rho_\theta} - \langle \vec{C}^\dagger_i \vec{C}_j \rangle_{\eta}.
\end{equation} 
We can solve this by first determining the target expectation values $\langle \vec{C}^\dagger_i \vec{C}_j \rangle_{\eta}$ and setting $\langle \vec{C}^\dagger_i \vec{C}_j \rangle_{\rho_\theta^*} = \langle \vec{C}^\dagger_i \vec{C}_j \rangle_{\eta}$. Then we use the fact that the Hamiltonian of a Gaussian Fermionic system can be written in the eigenbasis of the correlation matrix as 
\begin{equation}
     H_\eta =  \frac{1}{2} W_\eta \sigma^{-1}(\Lambda_\eta) W^\dagger_\eta,
\end{equation} 
where $W_\eta$ and $\Lambda_\eta$ is given by the eigendecomposition $\Gamma_{\eta} = W_\eta \Lambda_\eta W_\eta^\dagger$, and $\sigma^{-1}(X)$ the inverse sigmoid function. Thus, we (numerically) diagonalize $\Gamma_{\eta}$ and set the optimal Gaussian Fermionic QBM Hamiltonian equal to $H_{\theta^*}= \frac{1}{2}W_\eta \sigma^{-1}(\Lambda_\eta) W^\dagger_\eta$. Since the eigen decomposition of a Hermitian matrix is unique, we find one unique solution. This is in agreement with the strict convexity of the quantum relative entropy.

\subsubsection*{Geometrically-Local QBM}
\addcontentsline{toc}{subsubsection}{Geometrically-Local QBM}

The last type of restricted QBM model we discuss is the geometrically-local QBM. We focus on nearest-neighbor models on some $d$-dimensional lattice, for example a one-dimensional chain where each Pauli operator only acts on two neighboring qubits. In full generality, the geometrically-local QBM Hamiltonian is given by 
\begin{equation} 
\label{eq:nnmodel}
    \mathcal{H}_\theta = \sum_{k=x,y,z}\sum_{\langle i, j\rangle}\lambda^k_{ij}\sigma_i^k\sigma^k_j + \sum_i^n \gamma^k_i\sigma^k_i,
\end{equation} 
where the sum is over nearest-neighbor sites $\langle i,j \rangle$ of the lattice with periodic boundary conditions. In the main text, we consider for example a $d=1$ lattice (a ring), and a $d=2$ square lattice.

In order to use these models for pre-training, we train them with exact gradient descent on the relative entropy until a fixed precision is reached. Importantly, as these Hamiltonians only have $m \in \mathcal{O}(n)$ terms and a finite interaction range, Refs.~\cite{Anshu_2021, Haah_2024} show that the quantum relative entropy is strongly convex. Therefore, the optimization is guaranteed to converge quickly to the global optimum, recall Theorem~\ref{appthm:training_alpha}. However, this requires obtaining Gibbs state expectation values of geometrically local Hamiltonians. This can be done with a quantum computer or, potentially, with classical tensor networks~\cite{Kuwahara2021, Alhambra_2021}.

\section*{\MakeUppercase{Supplementary Note 5: Construction of the target state expectation values}}
\addcontentsline{toc}{section}{Supplementary Note 5: Construction of the target state expectation values}

In this appendix, we review how to embed classical data into a target density matrix $\eta$. We will follow the approach in~\cite{Kappen_2020} for quantum spin models. We also show how to extend this formalism to the Fermionic quantum models needed for the pre-training of our Gaussian Fermionic QBM. Lastly, we describe the two different targets used in our numerical simulations.

\subsection*{Classical data encoding}
\addcontentsline{toc}{subsection}{Classical data encoding}

Following~\cite{Kappen_2020}, one way to encode a classical data-set consisting of $N$ bit strings $\{\vec{s^\mu}\in\{0,1\}^n\}_{\mu=1}^M$ into a quantum state is by defining the pure state 
\begin{equation}
    \eta = \dyad{\psi},
\end{equation} 
with
\begin{equation}
    \ket{\psi} = \sum_{\vec{s}\in\{0,1\}^n} \sqrt{q(\vec{s})}\ket{\vec{s}}. 
\end{equation} 
Here $q(\vec{s})=\frac{1}{N}\sum_{\mu=1}^M\delta_{\vec{s}, \vec{s}^\mu}$ is the classical empirical probability for bitstring $\vec{s}$, and $\ket{\vec{s}}$ is a computational basis state indexed by $\vec{s}$. The $q(\vec{s})$ can be found by counting the bitstrings in the data set $\{\vec{s^\mu}\}$. From $\ket{\psi}$ one can compute expectation values such as 
\begin{equation} 
    \bra{\psi}\sigma^z_i\ket{\psi} = \sum_{\vec{s}\in\{0,1\}^n}q(\vec{s})\vec{s}_i
\end{equation} 
for the Pauli spin operator $\sigma^z_i$. This can be efficiently computed classically for a polynomially sized data set, i.e., for polynomially many $\vec{s^\mu}$. Computing such expectation values from $\eta$ is possible for all $1-$ and $2-$local Pauli operators as shown in Ref.~\cite{Kappen_2020}. 

We now show that we can generalize this encoding to Fermionic QBMs, i.e., the terms in the Hamiltonian ansatz consist of Fermionic creation $c_i^\dagger$ and annihilation operators $c_i$ . We define $\ket{\vec{s}}$ to be equal to the Fermionic Fock basis. This is analogous to the computational basis in the spin-picture (by the Jordan-Wigner transformation), but the bit-strings $\{\vec{s^\mu}\in\{0,1\}^n\}_{\mu=1}^M$ in the data set should now be interpreted as occupation-number vectors of Fermions. Note that the occupation-number basis is defined by the eigenstates of the Fermionic number operator $\sum_i c^\dagger_i c_i$. 

The creation and annihilation operators act on the Fock-basis states as follows 
\begin{equation} 
    c^\dagger_i \ket{\vec{s}} = (1-\vec{s}_i)\ket{\vec{s}+\vec{\delta_i}}, \hspace{1cm} c_i \ket{\vec{s}} = \vec{s}_i\ket{\vec{s}-\vec{\delta_i}},
\end{equation} 
where $\vec{\delta_i}$ is the unit bit-string with a $1$ at position $i$ and zeros everywhere else. With these relations, we can derive the required expectation values for the target $\eta$ to train the (Gaussian) Fermionic QBM   
\begin{align}
\begin{split}
    \bra{\psi} c^\dagger_i c_i \ket{\psi} &= \sum_{\vec{s}\in\{0,1\}^n}q(\vec{s})\vec{s}_i, \\
    \bra{\psi} c^\dagger_i c_j \ket{\psi} &= \sum_{\vec{s}\in\{0,1\}^n}\sqrt{q(\vec{s})q(F_i F_j\vec{s})}(1-\vec{s}_i)\vec{s}_j,
    \\
    \bra{\psi} c^\dagger_i c^\dagger_j \ket{\psi} &= \sum_{\vec{s}\in\{0,1\}^n}\sqrt{q(\vec{s})q(F_i F_j\vec{s})}(1-\vec{s}_i)(1-\vec{s}_j), \hspace{1cm} i\neq j
    \\
    \bra{\psi} c_i c_j \ket{\psi} &= \sum_{\vec{s}\in\{0,1\}^n}\sqrt{q(\vec{s})q(F_i F_j\vec{s})}\vec{s}_i\vec{s}_j, \hspace{1cm} i\neq j
    \\
    \bra{\psi} c_i\ket{\psi} &= \sum_{\vec{s}\in\{0,1\}^n}\sqrt{q(\vec{s})q(F_i \vec{s})}\vec{s}_i
    \\
    \bra{\psi} c_i^\dagger\ket{\psi} &= \sum_{\vec{s}\in\{0,1\}^n}\sqrt{q(\vec{s})q(F_i \vec{s})}(1-\vec{s}_i),
\end{split}
\end{align} 
where $F_i$ flips the Fermion occupation number (from occupied to unoccupied and vice versa) of index $i$ in the vector~$\vec{s}$. 

\subsection*{Data used for the numerical simulations}
\addcontentsline{toc}{subsection}{Data used for the numerical simulations}

For the numerical simulations, we use two different targets $\eta$: a target constructed from a quantum source, and a classical data set embedded into $\eta$ using the encoding above. For the quantum source we use the XXZ model Hamiltonian 
\begin{equation}
\label{eq:xxz}
    \mathcal{H}_{\mathrm{XXZ}} =\sum_{i=1}^{n-1}J\left(\sigma^{x}_i\sigma^{x}_{i+1}+\sigma^y_i\sigma^y_{i+1}\right)+\Delta \sigma_i^z\sigma^z_{i+1} + \sum_{i=1}^n h_z \sigma_i^z.
\end{equation} 
Here $J$ and $\Delta$ are parameters describing the Heisenberg interactions between the quantum spins on a one-dimensional lattice, and $h_z$ is the strength of an external magnetic field. We set $\eta=\frac{e^{\mathcal{H}_{\mathrm{XXZ}}}}{Z}$ with $J=-0.5$, $\Delta=-0.7$, and $h_z =- 0.8$, and compute the expectation values $\langle H_i\rangle_\eta$ classically. This is intractable in general, but our aim is to replicate the scenario in which the expectation values are measured experimentally. For example, from a state prepared on a quantum device. 

For the classical source, we use the classical salamander retina data set given in Ref.~\cite{Tkacik2014}. This data set consists of bit-string data recordings of different features of the response of cells in the salamander retina. We select the first $8$ features and trim the data to the first 10 data recordings. We then construct the expectation values $\langle H_i \rangle_\eta$ from the procedure outlined above.

\end{document}